\documentclass[11pt]{article}
\usepackage{amsmath}
\usepackage{amssymb}
\usepackage{amsthm}
\usepackage{latexsym}
\usepackage[dvips]{graphicx,color}
\usepackage{float}
\usepackage{array}
\usepackage{xcolor}
\usepackage{enumerate} 
\usepackage{titlesec}
\usepackage{multirow}
\usepackage[hyphens]{url}
\usepackage{hyperref}
\usepackage{etoolbox}
\usepackage{MnSymbol}
\usepackage{ifthen}
\usepackage{mathrsfs}
\usepackage{xspace}
\usepackage[normalem]{ulem}

\newtheoremstyle{mystyle}%
{3pt}
{3pt}
{\color{blue}}
{}
{\bfseries\color{blue}}
{.}
{.5em}
{}

\theoremstyle{plain}
\newtheorem{theorem}{Theorem}[section]

\newtheorem{prop}[theorem]{Proposition}
\newtheorem{corollary}[theorem]{Corollary}

\newtheorem{lemma}[theorem]{Lemma}

\theoremstyle{definition}

\newtheorem{definition}[theorem]{Definition}

\theoremstyle{remark}
\newtheorem{remark}[theorem]{Remark}

\theoremstyle{mystyle}

\newcommand{\fequiv}[1]{\ensuremath{\equiv_{#1, \fo}}}
\newcommand{\mequiv}[1]{\ensuremath{\equiv_{#1, \mso}}}
\newcommand{\lequiv}[1]{\ensuremath{\equiv_{#1, \mc{L}}}}

\newcommand{\true}{\ensuremath{\mathsf{True}}}

\newcommand{\fo}{\ensuremath{\text{FO}}}
\newcommand{\mso}{\ensuremath{\text{MSO}}}
\newcommand{\cmso}{\ensuremath{\text{CMSO}_1}}

\newcommand{\rank}[1]{\ensuremath{\text{rank}(#1)}}

\newcommand{\lt}{{{\L}o{\'s}-Tarski}}

\newcommand{\ls}{L\"owenheim-Skolem}
\newcommand{\dls}{downward L\"owenheim-Skolem}

\newcommand{\ldelta}[2]{\ensuremath{\Delta_{#1, \mc{L}, #2}}}

\newcommand{\eat}[1]{}

\newcommand{\tree}[1]{\ensuremath{\mathsf{#1}}}
\newcommand{\forest}[1]{\ensuremath{\mathsf{#1}}}

\newcommand{\cl}[1]{\ensuremath{\mathcal{#1}}}

\newcommand{\str}[1]{\ensuremath{\mathsf{Str}(#1)}}

\newcommand{\scale}[1]{\langle #1 \rangle}

\newcommand{\tbf}[1]{\textbf{#1}}
\newcommand{\mc}[1]{\mathcal{#1}}

\makeatletter
\newsavebox\myboxA
\newsavebox\myboxB
\newlength\mylenA

\newcommand*\xoverline[2][0.75]{%
    \sbox{\myboxA}{$\m@th#2$}%
    \setbox\myboxB\null
    \ht\myboxB=\ht\myboxA%
    \dp\myboxB=\dp\myboxA%
    \wd\myboxB=#1\wd\myboxA
    \sbox\myboxB{$\m@th\overline{\copy\myboxB}$}
    \setlength\mylenA{\the\wd\myboxA}
    \addtolength\mylenA{-\the\wd\myboxB}%
    \ifdim\wd\myboxB<\wd\myboxA%
       \rlap{\hskip 0.5\mylenA\usebox\myboxB}{\usebox\myboxA}%
    \else
        \hskip -0.5\mylenA\rlap{\usebox\myboxA}{\hskip 0.5\mylenA\usebox\myboxB}%
    \fi}
    
  \newcommand*\xunderline[2][0.75]{%
    \sbox{\myboxA}{$\m@th#2$}%
    \setbox\myboxB\null
    \ht\myboxB=\ht\myboxA%
    \dp\myboxB=\dp\myboxA%
    \wd\myboxB=#1\wd\myboxA
    \sbox\myboxB{$\m@th\underline{\copy\myboxB}$}
    \setlength\mylenA{\the\wd\myboxA}
    \addtolength\mylenA{-\the\wd\myboxB}%
    \ifdim\wd\myboxB<\wd\myboxA%
       \rlap{\hskip 0.5\mylenA\usebox\myboxB}{\usebox\myboxA}%
    \else
        \hskip -0.5\mylenA\rlap{\usebox\myboxA}{\hskip 0.5\mylenA\usebox\myboxB}%
    \fi}
\makeatother

\renewcommand{\str}[1]{\ensuremath{\mathcal{#1}}}
\newcommand{\tower}[1]{\ensuremath{\mathsf{tower}(#1)}}

\renewcommand{\scale}[2]{\ensuremath{\langle #1; #2 \rangle}}
\renewcommand{\forest}[1]{\ensuremath{\mathsf{#1}}}
\renewcommand{\root}[1]{\ensuremath{\mathsf{root}(#1)}}
\newcommand{\rootref}{\ensuremath{\mathsf{root}}}
\newcommand{\mdelta}[2]{\ensuremath{\delta_{#2}(#1)}}

\renewcommand{\ldelta}[2]{\ensuremath{\Delta^{\mc{L}}_{#1, #2}}}
\newcommand{\fdelta}[2]{\ensuremath{\Delta^{\fo}_{#1, #2}}}
\renewcommand{\mdelta}[2]{\ensuremath{\Delta^{\mso}_{#1, #2}}}
\newcommand{\lclass}[2]{\ensuremath{\delta^{\mc{L}}_{#1, #2}}}

\newcommand{\mclass}[2]{\ensuremath{\delta^{\mso}_{#1, #2}}}

\newcommand{\ef}{Ehrenfeucht-Fr\"aiss\'e}
\newcommand{\chaineq}{\ensuremath{\preceq}}
\newcommand{\alltypes}[1]{\ensuremath{\Delta_{#1}}}
\renewcommand{\mso}{\ensuremath{\mathrm{MSO}}}
\renewcommand{\fo}{\ensuremath{\mathrm{FO}}}
\newcommand{\gso}{\ensuremath{\mathrm{GSO}}}
\newcommand{\fpt}{\ensuremath{\mathrm{FPT}}}
\renewcommand{\rank}[1]{\ensuremath{\mbox{rank}(#1)}}
\renewcommand{\lequiv}[1]{\ensuremath{\equiv_{#1, \mc{L}}}}

\newcommand{\ind}[2]{\ensuremath{\mathcal{I}_{#1}(#2)}}
\renewcommand{\empty}{\ensuremath{\mbox{empty}}}
\newcommand{\elem}{\ensuremath{\mbox{elem}}}
\renewcommand{\cl}[1]{\ensuremath{\mathscr{#1}}}

\newcommand{\treemodel}[2]{\ensuremath{\mathrm{Tree}_{#1}(#2)}}
\newcommand{\tm}[2]{\ensuremath{\mathrm{TM}_{#1}(#2)}}

\newcommand{\lsref}{\ensuremath{\mathrm{LS}}}
\newcommand{\dlsref}{\ensuremath{\mathrm{DLS}}}
\newcommand{\ulsref}{\ensuremath{\mathrm{ULS}}}

\newcommand{\lelsref}{\ensuremath{\mc{L}\text{-}\mathrm{ELS}}}
\newcommand{\melsref}{\ensuremath{\mso\text{-}\mathrm{ELS}}}
\renewcommand{\ls}{\ensuremath{\text{L{\"o}wenheim-Skolem}}}
\renewcommand{\dls}{\ensuremath{\text{downward L{\"o}wenheim-Skolem}}}

\titleformat{\section}[block]{\Large\sc\filcenter}{\thesection.}{5pt}{}
\titleformat{\subsection}[block]{\sc\filcenter}{\thesubsection.}{5pt}{}

\makeatletter
\providecommand{\institute}[1]{
  \apptocmd{\@author}{\end{tabular}
    \par\smallskip
    \begin{tabular}[t]{c}
    #1}{}{}
}
\providecommand{\email}[1]{
  \apptocmd{\@author}{\end{tabular}
    \par\smallskip
    \begin{tabular}[t]{c}
    \texttt{#1}}{}{}
}
\makeatother

\begin{document}

\title{Some classical model theoretic aspects \\of bounded shrub-depth classes}
\author{Abhisekh Sankaran}

\institute{\small{Department of Computer Science and Technology,}\\\small{University of Cambridge, U.K.}}
\email{\small{\{abhisekh.sankaran\}@cl.cam.ac.uk}}
\date{}

\maketitle

\begin{abstract}
We consider classes of arbitrary (finite or infinite) graphs of bounded shrub-depth, specifically the class $\tm{r, p}{d}$ of $p$-labeled arbitrary graphs whose underlying unlabeled graphs have tree models of height $d$ and $r$ labels. We show that this class satisfies an extension of the classical {\ls} property into the finite and for $\mso$. This extension being a generalization of the small model property, we obtain that the graphs of $\tm{r, p}{d}$ are pseudo-finite. In addition, we obtain as consequences entirely new proofs of a number of known results concerning bounded shrub-depth classes (of finite graphs) and $\tm{r, p}{d}$. These include the small model property for $\mso$ with elementary bounds, the classical compactness theorem from model theory over $\tm{r, p}{d}$, and the equivalence of $\mso$ and $\fo$ over $\tm{r, p}{d}$ and hence over bounded shrub-depth classes. The proof for the last of these is via an adaptation of the proof of the classical Lindstr\"om's theorem from model theory characterizing $\fo$ over arbitrary structures.  
\end{abstract}

\section{Introduction}\label{section:intro}
\newcommand{\els}{\ensuremath{\mathrm{ELS}}}

Classical model theory is a subject that studies mathematical structures for their properties that can be expressed in a formal language such as
first order logic ($\fo$). The structures considered are arbitrary, so finite or infinite (but are typically infinite). Two of the earliest results in the subject, which form its pillars even today, are the L\"owenheim-Skolem theorem and the compactness theorem~\cite{chang-keisler}. The L\"owenheim-Skolem ($\lsref$) theorem states that for every infinite structure $\str{A}$ over a countable vocabulary and every infinite cardinal $\lambda$, there is a structure of size $\lambda$ that is $\fo$ equivalent to $\str{A}$, that further is a substructure of $\str{A}$ if $\lambda$ is less than the size of $\str{A}$ (the ``downward" part of the theorem, denoted $\dlsref$), and a superstructure of $\str{A}$ otherwise (the ``upward" part, denoted $\ulsref$).
On the other hand, the compactness theorem states that if every finite subset of an arbitrary set of $\fo$ sentences has a model, then the entire set has a model. The importance of these results can be gauged from a classic result by Lindstr\"om~\cite{lindstrom} who showed that for a reasonable notion of an abstract logic, if such a logic either satisfies both the $\dlsref$ and compactness theorems, or satisfies the full $\lsref$ theorem (so both its parts), then it is no more expressive than $\fo$.  

Shrub-depth is a structural parameter of classes of finite graphs that was introduced in~\cite{shrub-depth-intro} to admit fixed parameter tractable ($\fpt$) algorithms for a variety of interesting algorithmic problems, with an \emph{elementary} parameter dependence. The notion parameterizes classes of dense graphs by the height of their defining tree models, where informally, a tree model $\tree{t}$ of a graph $G$ is a rooted tree whose leaves are the vertices of $G$, and these are assigned labels from a finite set. The presence or absence of an edge between two vertices of $G$ is determined by the labels of these vertices in $\tree{t}$ and the distance between them in $\tree{t}$.  Since its inception in~\cite{shrub-depth-intro}, the notion has seen a lot of active research for not just its algorithmic properties, but also its structural and logical aspects~\cite{shrub-depth-FO-equals-MSO, shrub-depth-oum, CF20}. For instance, graphs of shrub-depth at most $d$ admit small models for $\mso$ sentences (of sizes bounded by elementary functions of the lengths of the sentences), and admit the equivalence of $\mso$ and $\fo$ in terms of their expressive powers~\cite{shrub-depth-FO-equals-MSO}.

In~\cite{abhisekh-csl17}, a finitary analogue of the {\dls} property was introduced. A class of structures has this property denoted $\mc{L}$-$\mathsf{EBSP}$ in~\cite{abhisekh-csl17} for a logic $\mc{L}$ such as $\fo$ or $\mso$, 
if any large structure in the class contains a small $\mc{L}[m]$-equivalent 
substructure that belongs to class. Here $\mc{L}[m]$-equivalent means that the two
structures agree on all $\mathcal{L}$ sentences of quantifier nesting depth (or rank) at most $m$, and ``small" means of size bounded by a computable function of $m$. Bounded shrub-depth classes satisfy  $\mso$-$\mathsf{EBSP}$  (which can be shown to be equivalent to the small submodel property for $\mso$), and it turns out going further, that a more general property holds of these classes, one called the \emph{logical fractal} property~\cite[Section 6]{abhisekh-csl17}. This property asserts ``logical self-similarity" at all ``scales" less than the size of a given structure, in analogy with the notion of fractals studied in the natural sciences. More specifically, for a strictly increasing function $f: \mathbb{N} \rightarrow \mathbb{N}$, called a ``scale function" in~\cite{abhisekh-csl17}, define the $i^{\text{th}}$ scale of $f$ as the set of all numbers in the interval $(f(i), f(i+1)]$. Then an $\mc{L}$-fractal is a class $\cl{C}$ of structures for which for each natural number $m$, there is a scale function $f_m$ such that for each structure $\str{A}$ in $\cl{C}$, there is an $\mc{L}[m]$-equivalent substructure $\str{B}$ of $\str{A}$ in $\cl{C}$, at each scale less than the scale of $\str{A}$. It turns out that bounded shrub-depth classes are $\mso$-fractals with computable scale functions. The $\mc{L}$-fractal property can be seen as a finitary $\mc{L}$-$\dlsref$ property that is closer to the classical $\dlsref$ property than $\mc{L}$-$\mathsf{EBSP}$. 

In this paper, we generalize significantly the above result for bounded shrub-depth classes, that further yields us various notable consequences. Specifically, we consider a finitary analogue of the full {$\lsref$} property, namely both its downward and upward parts, that therefore generalizes the logical fractal property. For $f$ as above, define for an infinite cardinal $\lambda$, the $\lambda^{\text{th}}$ scale of $f$ simply as the set $\{\lambda\}$. A class $\cl{C}$ of arbitrary structures (so now finite and infinite structures) is said to satisfy the \emph{extended $\mc{L}$-{\ls} ($\lelsref$) property} if there exists a scale function $f$ such that for any structure $\str{A}$ in $\cl{C}$, if $\str{A}$ is at scale $\eta$ of $f$ for some cardinal $\eta$ (that is, the size of $\str{A}$ belongs to the scale $\eta$), then for all cardinals $\lambda$, the following hold:
\vspace{1pt}\begin{enumerate}[1.]
    \item (Downward $\mc{L}$-$\els$): If $\lambda$ is at most  $\eta$, then (i) if $\lambda$ is finite, then there is a substructure of $\str{A}$, at scale $\lambda$ that is $\mc{L}[\lambda]$-equivalent to $\str{A}$, and (ii) if $\lambda$ is infinite, there is a substructure of $\str{A}$, at scale $\lambda$ that is $\mc{L}$-equivalent to $\str{A}$, where $\mc{L}$-equivalent means indistinguishable with respect to all sentences of $\mc{L}$.
    \item (Upward $\mc{L}$-$\els$): If $\lambda$ is at least $\eta$, then (i) if $\eta$ is finite, then there is a superstructure of $\str{A}$ at scale $\lambda$ that is $\mc{L}[\eta]$-equivalent to $\str{A}$, and (ii) if $\eta$ is infinite, then there is a superstructure of $\str{A}$ at scale $\lambda$ that is $\mc{L}$-equivalent to $\str{A}$.
\end{enumerate}
We call $f$ a ``witness function" for the  $\mso$-$\els$ property of $\cl{C}$.
From the perspective of fractals, the $\lelsref$ property can be seen as asserting ``decreasing $\mc{L}$-self-similarity" going downwards from the size of a given structure and ``increasing $\mc{L}$-self-similarity" going upwards. It is not difficult to see that for any given $m$, by adjusting the scales suitably, specifically by merging all scales of $f$ that are at most $m$, the $\lelsref$ property, restricted to the finite, implies the $\mc{L}$-fractal property since $\mc{L}[m']$-equivalence implies $\mc{L}[m]$-equivalence for all $m' \ge m$. 
We now consider the extension of the notion of bounded shrub-depth to arbitrary graphs as introduced in~\cite{CF20}, obtained by considering tree models of bounded height but arbitrary branching. Specifically, for $r, p, d \in \mathbb{N}$, we consider the class $\tm{r, p}{d}$ of arbitrary $p$-labeled graphs (that is, graphs whose vertices are labeled with labels from $\{1, \ldots, p\}$) that admit arbitrary tree models of height $d$ and $r$ labels, and show as the central result of our paper, the following theorem. 

\begin{theorem}\label{thm:main-display}
The class $\tm{r, p}{d}$ satisfies the $\mso$-$\els$ property with elementary ($(d+1)$-fold exponential) witness functions.
\end{theorem}

We prove the above theorem in two parts, the first where one of $\lambda$ or $\eta$ in the $\melsref$ definition is finite (Theorem~\ref{thm:shrub-depth-full-LS-finite}), and the second when both are infinite (Proposition~\ref{prop:shrub-depth-full-LS-infinite}). Both of these parts crucially build upon the result that $\tm{r, p}{d}$ admits small $\mso[m]$-equivalent submodels whose sizes are bounded by a $d$-fold exponential function of $m$ (Theorem~\ref{thm:shrub-depth-DLS-and-index-bound:DLS}). Thus in the parlance of model theory, the graphs of $\tm{r, p}{d}$ are  \emph{pseudo-finite}, in that every $\fo$ (indeed also $\mso$) sentence true in a graph $G$ of $\tm{r, p}{d}$, is also true in a finite graph of $\tm{r, p}{d}$, and further even in infinitely many finite induced subgraphs of $G$ if $G$ is infinite.

As consequences of Theorem~\ref{thm:main-display}, we obtain entirely new proofs for a variety of known results in the literature concerning bounded shrub-depth classes. Firstly, it follows that shrub-depth $d$ classes of finite graphs admit small models for $\mso$ of elementary sizes (shown in~\cite{shrub-depth-FO-equals-MSO}); consequently, the index of the $\mso[m]$ equivalence relation over any shrub-depth $d$ class is elementary again. Secondly, we obtain a new proof of the classical compactness theorem from model theory mentioned at the outset, over the class $\tm{r, p}{d}$. That is, we show that for any set $T$ of $\mso$ sentences over the vocabulary of $p$-labeled graphs, if every finite subset of $T$ has a model in $\tm{r, p}{d}$, then all of  $T$ has a model in $\tm{r, p}{d}$; and further, a countable such model (Theorem~\ref{thm:shrub-depth-compactness}). It can be seen following~\cite{CF20} that $\tm{r, p}{d}$ can be axiomatized (over all structures) by a single $\fo$ sentence, so that $\tm{r, p}{d}$ is an \emph{elementary} class in the classical model-theoretic language, where elementary here means ``axiomatizable by a set of $\fo$ sentences". Since most classical model theoretic results, and in particular compactness, hold over elementary classes, the $\fo$ compactness theorem for $\tm{r, p}{d}$ follows already from literature. However the standard proofs of the compactness theorem are either via G\"odel's completeness theorem for $\fo$, or Henkin models, or Skolemization or ultraproducts~\cite{chang-keisler}. Our approach is completely different from all of these and goes via using the pseudo-finiteness of the graphs of $\tm{r, p}{d}$ to construct a chain under induced subgraph of models of increasingly larger subsets of the given theory $T$, whose union is a countable model of $T$. We believe this technique might be of independent interest. 

The reader might wonder why we have not mentioned our $\mso$ compactness theorem for $\tm{r, p}{d}$ to be more general than the $\fo$ compactness theorem. The reason is that, as it turns out, $\mso$ is no more expressive than $\fo$ over $\tm{r, p}{d}$ (Theorem~\ref{thm:mso=fo}). This fact has been shown for $\tm{r, 1}{d}$ (so basically the unlabeled versions of graphs in $\tm{r, p}{d}$) in~\cite{CF20}, and for the class of finite graphs in $\tm{r, 1}{d}$ in~\cite{shrub-depth-FO-equals-MSO}, but we provide a novel proof of this result (and further also for $\tm{r, p}{d}$ for all $p > 1$) by adapting the ideas from the proof of the classical Lindstr\"om's theorem mentioned at the outset, which indeed asserts a "logical collapse" to $\fo$ for any abstract logic having  the {\dlsref} and compactness properties. Since $\mso$ does satisfy the conditions of the mentioned abstract logic, it is natural to think that Lindstr\"om's theorem can be employed as is, to prove the $\mso$-$\fo$ equivalence over $\tm{r, p}{d}$. But it turns out that cannot be done since the proof of Lindstr\"om's theorem requires a compactness theorem for structures over the given vocabulary expanded with binary relation and function symbols. Therefore, we directly adapt and implement the ideas from Lindstr\"om's proof, in our setting. We remark that our $\mso$ compactness theorem, particularly the fact that it gives a \emph{countable} model, turns out to be vital for this adaptation to go through. 

The main tool at the heart of our results, is an ``unbounded" version of  Feferman-Vaught composition theorem proved by Elberfeld, Grohe and Tantau in~\cite{FO-MSO-coincide}. The standard Feferman-Vaught composition theorem shows that evaluating an $\mso$ sentence over the disjoint union of finitely many structures is equivalent to evaluating a finite set of $\mso$ sentences over each of the individual structures, and then putting together the results of the evaluations using a Boolean propositional logic formula. In~\cite{FO-MSO-coincide}, this result is extended elegantly to infinitely many structures using $\fo$ instead of propositional logic. Specifically, for any $\mso$ sentence $\Phi$, to know whether it is true over the disjoint union of an arbitrary family $\mc{F}$ of structures, it is equivalent to examine the truth of an $\fo$ sentence $\alpha_\Phi$ over an \emph{$\mso[m]$-type indicator} $\ind{m}{\mc{F}}$ which is a structure over a monadic vocabulary, that contains all the information about the equivalence classes of the $\mso[m]$ relation to which the structures in $\mc{F}$ belong (indeed these classes constitute the mentioned vocabulary). The $\fo$ sentence $\alpha_\Phi$ constructed inductively on the structure of $\Phi$, can be seen as doing both the job of evaluating suitable $\mso[m]$ sentences over the individual structures of $\mc{F}$ and the job of compiling together the results of the individual evaluations. The key observation for our results is the fact that the vocabulary of the type indicator is \emph{monadic}. Structures over monadic vocabularies are particularly nice when it comes to dealing with them using $\fo$ sentences of bounded rank, since, when they are sufficiently large, one can always shrink them (upto a certain threshold depending on the rank) or expand them (without any bounds, even infinitely) by respectively deleting or adding one element at a time without changing their $\fo[q]$ theory, where $q$ is the considered rank. This simple fact is exploited over and over again to prove all our results mentioned above. Specifically, the results are first proved for trees of bounded height noting that trees are after all constructed inductively from forests of lesser height, and the latter lend themselves to using the mentioned composition theorem. Subsequently, these results for trees are transferred to $\tm{r, p}{d}$ using the $\fo$ interpretability of the latter in the former.

The organization of the paper is as follows. In Section~\ref{section:prelims}, we provide the background and notation for the paper, also recall the Feferman-Vaught composition theorem from~\cite{FO-MSO-coincide}. In Section~\ref{section:elementary-bounds}, we show the small model property for $\mso$ over $\tm{r, p}{d}$ with elementary bounds. In Section~\ref{section:ELS}, we show the $\mso$-$\els$ property for $\tm{r, p}{d}$. In Section~\ref{section:compactness} we prove the $\mso$ compactness theorem over $\tm{r, p}{d}$ and utilize this in Section~\ref{section:mso=fo} to show the $\mso$-$\fo$ equivalence over $\tm{r, p}{d}$. We present our conclusions in Section~\ref{section:conclusion}.

\section{Background}\label{section:prelims}

We assume the reader is familiar with the terminology and notation in connection with $\mso$. A vocabulary, typically denoted $\tau$ or $\sigma$, is a finite set of relation symbols. We denote the class of all $\mso$ formulae over $\tau$ as $\mso(\tau)$ and its subclass of $\fo$ formulae as $\fo(\tau)$. A sequence $x_1, \ldots, x_n$ of $\fo$ variables is denoted $\bar{x}$. An $\mso$ formula $\varphi$ whose free variables are in $\bar{x}$ is denoted $\varphi(\bar{x})$. A sentence is a formula without free variables. The \emph{quantifier rank}, or simply rank, of an $\mso$ formula $\varphi$, denoted $\rank{\varphi}$, is the maximum number of quantifiers (both first order and second order) appearing in any root to leaf path in the parse tree of the formula. We denote by $\mso[m]$ the class of all $\mso$ formulae of rank at most $m$. 

A $\sigma$-structure $\str{A}$  consists of a universe $A$ equipped with interpretations of the predicates of $\sigma$. We denote by $\str{B} \subseteq \str{A}$ that $\str{B}$ is a substructure of $\str{A}$, and by $\str{A} \cong \str{B}$ that $\str{A}$ and $\str{B}$ are isomorphic. We denote by $\mathbb{N}$ the set of all natural numbers, by $\mathbb{N}_+$ the set  $\mathbb{N} \setminus \{0\}$, and by $[p]$ the set $\{1, \ldots, p\}$ for any $p \in \mathbb{N}_+$. For $p \in \mathbb{N}_+$, let $\tau = \sigma \cup \{P_1, \ldots, P_p\}$ be a vocabulary obtained by expanding $\sigma$ with $p$ new unary relation symbols $P_1, \ldots, P_p$. A \emph{$p$-labeled} $\sigma$-structure is a $\tau$-structure $\str{A}$ in which for every element $a$ in $A$, there exists a unique $i \in [p]$ such that $a \in P_i^{\str{A}}$. (It is allowed for $P_j^{\str{A}}$ to empty for some $j \in [p]$.) We say of the mentioned element $a$, that it is labeled with the label $i$. We consider simple, undirected and loop-free graphs in this paper, where such a graph $G$ is a $\sigma$-structure for $\sigma = \{E\}$, in which the binary relation symbol $E$ is interpreted as an irreflexive and symmetric relation. The graphs can have \emph{arbitrary} cardinality (i.e. whose universe could be finite or infinite). We denote by $V(G)$ and $E(G)$, the vertex and edge sets of $G$. As above,  a \emph{$p$-labeled graph} $G$ is graph each of whose vertices is assigned a unique label from $[p]$. If $G_1$ and $G_2$ are two $p$-labeled graphs, and $G_1 \subseteq G_2$, then we say $G_1$ is a \emph{labeled induced subgraph} of $G_2$.

A \emph{tree} is a connected graph that does not contain any cycles. A rooted tree $\tree{t}$ is a $\{E, \rootref\}$-structure whose $\{E\}$-reduct is a tree and in which the unary relation symbol $\rootref$ is interpreted as a set consisting of a single element called the \emph{root} of $\tree{t}$, and denoted $\root{\tree{t}}$. We can now talk of a \emph{$p$-labeled rooted tree} as a rooted tree in which each node (including the root) is assigned a unique label from $[p]$. We shall often call $p$-labeled rooted trees, as simply trees when $p$ is clear from context. For a tree $\tree{t}$ and a node $v$ of it, the \emph{subtree of $\tree{t}$ rooted at $v$}, denoted $\tree{t}_v$, is the $\{E, \rootref\}$-structure such that (i) the $\{E\}$-reduct of $\tree{t}_v$ is the substructure of $\tree{t}$ induced by the set of all nodes $u$ of $\tree{t}$ for which $v$ lies on the path from $u$ to $\root{\tree{t}}$, and (ii) $\rootref$ is interpreted in $\tree{t}_v$ as the set $\{v\}$. We say a tree $\tree{s}$ is a subtree of $\tree{t}$ if it is the subtree of $\tree{t}$ rooted at some node $v$ of $\tree{t}$ (note the slight abuse of the standard `subtree' terminology, intended for later convenience). If $v$ is a child of $\root{\tree{t}}$, then we say $\tree{t}_v$
is a \emph{child subtree} of $\root{\tree{t}}$ in $\tree{t}$.  
A $p$-labeled rooted forest $\forest{f}$ is a disjoint union of $p$-labeled rooted trees. If $(\tree{t}_i)_{i \in I}$ for an index set $I$ is the family of $p$-labeled rooted trees constituting $\forest{f}$, then we write $\forest{f} = \bigcupdot_{i \in I} \tree{t}_i$ where $\bigcupdot$ denotes disjoint union. A subtree $\tree{t}_2$ of a ($p$-labeled rooted) tree $\tree{t}_1$ is said to be \emph{leaf-hereditary} if the roots of $\tree{t}_2$ and $\tree{t}_1$ are the same, and every leaf of $\tree{t}_2$ is also a leaf of $\tree{t}_1$. It follows that if $\forest{f}_i$ is the forest of $p$-labeled rooted trees obtained by removing the root of $\tree{t}_i$ for $i \in \{1, 2\}$, then for every tree $\tree{s}_2$ of $\forest{f}_2$, there exists a tree $\tree{s}_1$ of $\forest{f}_1$ such that $\tree{s}_2$ is a leaf-hereditary subtree of $\tree{s}_1$. The \emph{height} of a tree $\tree{t}$ is the maximum root to leaf distance in $\tree{t}$. The height of a node $v$ of a tree $\tree{t}$ is the height of the subtree $\tree{t}_v$. Singleton trees have height 0. We denote by $\cl{T}_{d, p}$ the class of all (arbitrary cardinality) $p$-labeled rooted trees of height at most $d$. 

\subsection{Shrub-depth}\label{section:shrub-depth}
We recall the notion of tree models from~\cite{shrub-depth-definitive} and state it in its extended version for arbitrary cardinality graphs. For $r, d \in \mathbb{N}$, a \emph{tree model of $r$ labels and height $d$} for a graph $G$ is a pair $(\tree{t}, S)$ where $\tree{t}$ is an $(r+1)$-labeled rooted tree of height $d$ and $S \subseteq [r]^2 \times [d]$ is a set called the \emph{signature} of the tree model such that:
\vspace{2pt}\begin{enumerate}
    \item The length of every root to leaf path is exactly $d$.
    \item The set $V(G)$ is exactly the set of leaves of $\tree{t}$.
    \item Each leaf is assigned a unique label from $[r]$ and all internal nodes are labeled $r+1$.
    \item For any $i, j \in [r]$ and $l \in [d]$, it holds that $(i, j, l) \in S$ if and only if $(j, i, l) \in S$. 
    \item For vertices $u, v \in V(G)$, if $i$ and $j$ are the labels of $u$ and $v$ seen as leaves of $\tree{t}$, and the distance between $u$ and $v$ in $\tree{t}$ is $2l$, then $\{u, v\} \in E(G)$ iff $(i, j, l) \in S$. Observe that the distance between $u$ and $v$ in $\tree{t}$ is an even number as all root to leaf paths are of length $d$, and $l$ is thus the distance between $u$ (or $v$) and the least common ancestor of $u$ and $v$.
\end{enumerate}

The class of all (arbitrary) tree models of $r$ labels and height $d$ is denoted $\treemodel{r}{d}$ and the class of all (arbitrary) graphs that have tree models in $\treemodel{r}{d}$ is denoted $\tm{r}{d}$. A class $\cl{C}$ of arbitrary graphs is said to have \emph{shrub-depth} $d$ if $\cl{C} \subseteq \tm{r}{d}$ for some $r \ge 1$, and $\cl{C} \not\subseteq \tm{r'}{d-1}$ for any $r' \ge 1$. We say $\cl{C}$ has bounded shrub-depth if $\cl{C} \subseteq \tm{r}{d}$ for some $r, d \ge 1$. For the purposes of our results, we extend the notion of tree models recalled above in a simple way to handle classes of $p$-labeled graphs whose underlying unlabeled graphs form classes having bounded shrub-depth. Given $r, p, d \in \mathbb{N},  p \ge 1$, let $\tree{s}$ be a tree whose internal nodes are labeled $r+1$ and each of whose leaf nodes is labeled with a pair $(i, j)$  where $i \in [r]$ and  where $j \in [p]$. For a set $S \subseteq [r]^2 \times [d]$, we say $(\tree{s}, S) \in \treemodel{r, p}{d}$ if the tree $\tree{t}$ obtained from $\tree{s}$ by replacing the pair labeling any leaf node with the first component of the pair (so replacing every leaf label of the form $(i, j)$ with just $i$), is such that $(\tree{t}, S) \in \treemodel{r}{d}$. It is easy to see that $(\tree{s}, S)$ defines a $p$-labeled graph $G$ whose underlying unlabeled graph is the graph of $\tm{r}{d}$ for which $(\tree{t}, S)$ is a tree model, and such that for any node $u \in V(G)$, if $(i, j)$ is the label of $u$ seen as a leaf node of $\tree{s}$, then the label of $u$ in $G$ is $j$. We say that $(\tree{s}, S)$ is a tree-model of $G$. Denote by $\tm{r, p}{d}$ the class of  arbitrary  $p$-labeled graphs $G$ for which there is a tree model in $\treemodel{r, p}{d}$.  It is easy to see that $\tm{r, p}{d}$ is a hereditary class, i.e. it is closed under induced subgraphs.

From the definitions above, we see that for any $(\tree{s}, S) \in \tm{r, p}{d}$, the tree $\tree{s}$ can be seen as belonging to $\cl{T}_{d, r \cdot p+1}$ via the bijective function $f: [r] \times [p] \rightarrow [r \cdot p]$ that maps the pair $(i, j)$ labeling any leaf node of $\tree{s}$ to the number $(i-1) \cdot p + j$, with every internal node of $\tree{s}$ getting the label $r \cdot p + 1$. We will therefore also use the notation $\tree{s} \in \cl{T}_{d, r \cdot p  + 1}$. It is easy to write an $\fo$ sentence $\Omega_{r, p, d}$ over the vocabulary of $\cl{T}_{d, r \cdot p + 1}$ that says that for a given tree from $\cl{T}_{d, r \cdot p + 1}$, every root to leaf path has length exactly $d$, every internal node is labeled with the label $r \cdot p + 1$, and no leaf node is labeled with the label $r\cdot p + 1$. Then any model $\tree{s} \in \cl{T}_{d, r \cdot p + 1}$ of $\Omega_{r, p, d}$ when equipped with a signature $S \subseteq [r]^2 \times [p]$  gives a tree model $(\tree{s}, S) \in \treemodel{r, p}{d}$, and conversely for any tree model $(\tree{s}, S) \in \treemodel{r, p}{d}$, we have $\tree{s}$, seen as a tree of $\cl{T}_{d, r\cdot p +1}$, models $\Omega_{r, p, d}$. Observe that the rank of $\Omega_{r, p, d}$ is $O(d)$. For every signature $S \subseteq [r]^2 \times [p]$, there exists a tuple $\Xi_{S, p}$ of $\fo$ formulas given by $\Xi_{S, p} = (\xi_{V, S}(x), \xi_{E, S}(x, y), (\xi_{P_j, S}(x))_{1 \leq j \leq p})$ that when evaluated on a tree $\tree{s}$ satisfying $\Omega_{r, p, d}$ produces the graph $G \in \tm{r, p}{d}$ of which $(\tree{s}, S) \in \treemodel{r, p}{d}$ is a tree model. Specifically, (i) the formula $\xi_{V, S}(x)$ says that $x$ is a leaf node in $\tree{s}$, (ii) the formula $\xi_{E, S}(x, y)$ says that $x$ and $y$ are leaves of $\tree{s}$ and for some $(i, j, l) \in S$, it holds that $i$ and $j$ are resp. the labels of $x$ and $y$, and the distance between $x$ and $y$ in $\tree{t}$ is exactly $2l$, and (iii) the formula $\xi_{P_j, S}(x)$ says that $x$ is a leaf node whose label corresponds to the pair $(i, j)$ for some $i \in [r]$. We call $\Xi_{S, p}$ an \emph{$\fo$ interpretation} of $\tm{r, p}{d}$ in $\treemodel{r, p}{d}$. Thus $\Xi_{S, p}$ defines a function from  $\treemodel{r, p}{d}$ to $\tm{r, p}{d}$, which we also denote as $\Xi_{S, p}$; so for $G$ and $\tree{s}$ as above $G = \Xi_{S, p}(\tree{s})$. Then $\tm{r, p}{d} = \bigcup_{S \subseteq [r]^2 \times [p]} \{\Xi_{S, p}(\tree{s}) \mid (\tree{s}, S) \in \treemodel{r, p}{d}\}$. We now mention some important observations about $\Xi_{S, p}$. Firstly, the rank of $\Xi_{S, p}$, defined as the maximum rank of the formulae appearing in it, is $O(d)$. Next, if $G = \Xi_{S, p}(\tree{s})$, then for any $\mso$ formula $\varphi$ in the vocabulary of $G$, there is an $\mso$ formula that we denote $\Xi_{S, p}(\varphi)$ in the vocabulary of $\tree{s}$ such that $G \models \varphi$ iff $\tree{s} \models \Xi_{S, p}(\varphi)$. (This is a special case of a more general result called the fundamental theorem of interpretations.) As a consequence, if $\tree{s}_1$ and $\tree{s}_2$ are two models of $\Omega_{r, p, d}$ and $q = \mbox{rank}(\Xi_{S, p})$ for a signature $S$, then $\tree{s}_1 \mequiv{m+q} \tree{s}_2$ implies $\Xi_{S, p}(\tree{s}_1) \mequiv{m} \Xi_{S, p}(\tree{s}_2)$. Finally for $(\tree{s}, S) \in  \treemodel{r, p}{d}$, if $\tree{s}'$ is a leaf-hereditary subtree of $\tree{s}$, then $\Xi_{S, p}(\tree{s}')$ is a labeled induced subgraph of $\Xi_{S, p}(\tree{s})$.

\subsection{Feferman-Vaught composition}\label{section:FV-comp} 
Let $\mc{L}$ be one of the logic $\fo$ or $\mso$. Given $m \in \mathbb{N}$ and a $\tau$-structures $\str{A}$ and $\str{B}$, we say that $\str{A}$ and $\str{B}$ are $\mc{L}[m]$-equivalent, denoted $\str{A} \lequiv{m} \str{B}$, if $\str{A}$ and $\str{B}$ agree on all $\mc{L}$ sentences of rank at most $m$. The relation $\mequiv{m}$ is an equivalence relation of finite index. Given a class $\cl{C}$ of $\tau$-structures and $m \in \mathbb{N}$, we let $\ldelta{m}{\cl{C}}$ denote the set of all equivalence classes of the $\lequiv{m}$ relation over $\cl{C}$. We denote by $\lclass{m}{\cl{C}}(\str{A})$ the equivalence class of $\ldelta{m}{\cl{C}}$ that contains $\str{A}$. For every $\delta \in \ldelta{m}{\cl{C}}$, there exists an $\mc{L}[m]$ sentence $\Theta_\delta$ that defines $\delta$ over $\mc{C}$. Define the function $\mathsf{tower}: \mathbb{N} \times \mathbb{N} \rightarrow \mathbb{N}$ as: $\tower{0, n} = n$ and $\tower{d, n} = 2^{\tower{d-1, n}}$. Let $\iota(m, \cl{C}, \mc{L})$ denote the (finite) index of the $\lequiv{m}$ relation over $\cl{C}$; so $\iota(m, \cl{C}, \mc{L}) = |\ldelta{m}{\cl{C}}|$. The relation $\lequiv{m}$ has a characterization in terms of Ehrenfeucht-Fra\"iss\'e games. We point the reader~\cite[Chapters 3 \& 7]{libkin} for the details concerning these. 

Let $\tau$ and $\cl{C}$ be as above. Let $\mc{F} = (\str{A}_i)_{i \in I}$ be a family of structures of $\cl{C}$ with disjoint universes, indexed by an index set $I$ of an arbitrary cardinality. Let $m \in \mathbb{N}$ and $\tau_{m,  \cl{C}}$ be the relational vocabulary consisting of a distinct unary predicate symbol $T$ for each class $T \in \mdelta{m}{\cl{C}}$, and containing no other predicate symbols. Note that we are using the same symbol $T$ to denote both an equivalence class of $\mequiv{m}$ over $\cl{C}$, as well as the unary predicate symbol corresponding to the class for the ease of understanding; whether we mean $T$ as a class or a predicate symbol will be clear from context.  The \emph{$\mso[m]$-type indicator} for the family $\mc{F}$ is now defined as a $\tau_{m,  \cl{C}}$-structure $\ind{m}{\mc{F}}$ such that (i) the universe of $\ind{m}{\mc{F}}$ is $I$, and (ii) for $T \in \mdelta{m}{\cl{C}}$, the interpretation of $T$ in $\ind{m}{\mc{F}}$ is the set $\{i \in I \mid T = \mclass{m}{\cl{C}}(\str{A}_i)\}$. Observe that for each $i \in I$, there is exactly one predicate $T \in \tau_{m,  \cl{C}}$ such that  $i$ is in the interpretation of $T$ in $\ind{m}{\mc{F}}$. We now have the following theorem from~\cite{FO-MSO-coincide}. (This is the  special case of $w = 0$ and $\mathrm{L} = \mso$ in~\cite[Theorem 14]{FO-MSO-coincide}.)


\begin{theorem}[Theorem 14,~\cite{FO-MSO-coincide}]\label{thm:FV-decomposition}
Let $\cl{C}$ be a class of structures over a vocabulary $\tau$. For every $\mso$ sentence $\Phi$ over $\tau$ of rank $m$, there exists an $\fo$ sentence $\alpha_{\Phi}$ over $\tau_{m, \cl{C}}$ such that if $\mc{F} = (\str{A}_i)_{i \in I}$ is a family of structures of $\cl{C}$ with disjoint universes, indexed by an index set $I$ of an arbitrary cardinality, then the following holds:
\[
\ind{m}{\mc{F}} \models \alpha_\Phi ~~~\mbox{if, and only if,}~~~ \bigcupdot_{i \in I} \str{A}_i \models \Phi
\]
Further, if $\cl{\widetilde{C}}$ represents the class of structures of $\cl{C}$ expanded with (all possible interpretations of) $m$ new unary predicate symbols, then the rank of $\alpha_\Phi$ is $O((\iota(m, \cl{\widetilde{C}}, \mso))^{m+1})$.
\end{theorem}

\begin{remark}
In~\cite{FO-MSO-coincide}, the result is actually stated for $\Phi$ which does not contain any $\fo$ variables and whose atomic formulae, instead of being the usual atomic formulae (of the form $x_1 = x_2$, $Y(x)$ for an $\mso$ variable $Y$, and $R(x_1, \ldots, x_r)$ where $R \in \tau$ and $x, x_1, \ldots, x_r$ are $\fo$ variables), are instead of the ``second order" forms $\empty(X)$ and $\elem(X_1, \ldots, X_r, Z)$ where $X, X_1, \ldots, X_r$ are $\mso$ variables and $Z$ is either an $\mso$ variable or a predicate of $\tau$, with $r$ being the arity of $Z$. The semantics for these atomic forms is intuitive: $\empty(X)$ holds for set $P$ if $P$ is empty, and $\elem(X_1, \ldots, X_r, Z)$ holds for $X_i = P_i$ and $Z = Q$ if $|P_i| = 1$ for $1 \leq i \leq r$ and $P_i \times \cdots \times P_r \subseteq Q$. Every "usual" $\mso$ formula can be converted into an equivalent $\mso$ formula over the mentioned second order atomic formulae, without any change of quantifier rank (see~\cite[page 4]{FO-MSO-coincide}). We have hence recalled~\cite[Theorem 14]{FO-MSO-coincide} in the form stated above in Theorem~\ref{thm:FV-decomposition} which features $\Phi$ as a usual $\mso$ formula, since these are the kinds of formulae which we work with in this paper. 
\end{remark}

We provide here the justification for the last statement of Theorem~\ref{thm:FV-decomposition} which does not appear explicitly in~\cite{FO-MSO-coincide} but is indeed a consequence of the proof of~\cite[Theorem 14]{FO-MSO-coincide}. We refer the reader to~\cite[Section 3.1, pp. 6 -- 9]{FO-MSO-coincide} to find the formulae and other constructions we refer to in our description here. 

We first observe in the proof of Lemma 8 of~\cite{FO-MSO-coincide}, that the ``capping" constant $C$ for the formula $\alpha$ is simply the rank of $\alpha$, since $\alpha$ is an FO sentence over a monadic vocabulary $\tau_\Sigma$ (and we have also seen a similar such result in Lemma~\ref{lemma:monadic-structures}). Then the number $n$  mentioned in the proof is at most $\mbox{rank}(\alpha) \cdot |\tau_\Sigma|$, whereby the rank of $\beta_z$, and hence the rank of $\beta_\alpha$, is at most $\mbox{rank}(\alpha) \cdot |\mbox{vocab}(\alpha)| + 1$, where $\mbox{vocab}(\alpha)$  denotes the vocabulary of $\alpha$, namely $\tau_\Sigma$. Call this observation (*).

We now come to the proof of Theorem 14 of~\cite{FO-MSO-coincide}, and make the following observations about the rank of $\alpha_\Phi$ following the inductive construction of $\alpha_\Phi$ as given in the proof. For the base cases, since $\gamma^{\mathrm{L}}_\Psi(i)$ is a quantifier-free formula for any $\Psi$, we get that if $\Phi := \empty(X)$, then rank of $\alpha_\Phi$ is 1, and if $\Phi := \elem(X_1, \ldots, X_r, Z)$ or $\Phi := \elem(X_1, \ldots, X_r, R)$, then the rank of $\alpha_\Phi$ is 2 since the width $w$ is 0 by our assumption. If $\Phi$ is a Boolean combination of a set of formulae, then $\mbox{rank}(\alpha_\Phi)$ is the maximum of the ranks of the formulae in the mentioned set. We now come to the non-trivial case when $\Phi := \exists X (\Phi')$. 

We see from~\cite[page 9, para 2]{FO-MSO-coincide} that the rank of $\alpha_\Phi$ is the maximum of the ranks of $\alpha_C$ where $\alpha_C$ is obtained from $\beta_{\alpha_{\Phi'}}$ (denoted as simply $\beta$ in the proof as a shorthand) by substituting the atoms $T(i)$ with the quantifier-free formula $\gamma^{\mathrm{L}}_\Psi(i)$ for a suitable $\Psi$. Then $\mbox{rank}(\alpha_\Phi) = \mbox{rank}(\beta_{\alpha_{\Phi'}})$. It follows from (*) above that $\mbox{rank}(\beta_{\alpha_{\Phi'}}) \leq \mbox{rank}(\alpha_{\Phi'}) \cdot |\mbox{vocab}(\alpha_{\Phi'})| + 1$; call this inequality (**).  
Let $\equiv_{r, \mathrm{L}}$ denote the equivalence relation that relates two structures over the same vocabulary iff they agree on all $\mathrm{L}$ sentences (over the vocabulary of the structures) of rank at most $r$. Then the vocabulary of $\alpha_{\Phi'}$ is the set of all equivalence classes of the $\equiv_{q - 1, \mathrm{L}}$ relation over all structures over the vocabulary of the family $F$ (where $F$ is as in the statement of~\cite[Theorem 14]{FO-MSO-coincide}), expanded with (all possible interpretations of) $d$ set predicates, where $d$ is the number of free variables of $\Phi'$ and $q - 1$ is the rank of $\Phi'$. Here we now importantly observe that if the structures of the family $F$ come from a class $\cl{C}$, then it is sufficient to consider just those equivalence classes of the $\equiv_{q - 1, \mathrm{L}}$ relation that are non-empty when restricted to the structures of $\cl{C}$ expanded with $d$ set predicates. Then if $\Phi'$ is a subformula of a rank $m$ $\mathrm{L}$ sentence $\Phi$ over a vocabulary $\tau$ and we are interested only in a given class $\cl{C}$ of $\tau$-structures and expansions of these with set predicates, then the size of the vocabulary of $\alpha_{\Phi'}$ is at most the index of the $\equiv_{m, \mathrm{L}}$ equivalence relation over the class $\widetilde{\cl{C}}$ of structures of $\cl{C}$ expanded with (all possible interpretations of exactly) $m$ set predicates. Applying this observation to (**) iteratively, we then get that if $\mathrm{L} = \mso$ and $\lambda = \iota(m, \cl{\widetilde{C}})$ then
\[
\begin{array}{lll}
\mbox{rank}(\alpha_\Phi) &  \leq & 1 + \lambda \cdot \Big(1 + \lambda \cdot \big( 1 + \ldots \cdot (1 + 2 \cdot \lambda)\big)\Big)\\
     & \leq &  1 + \lambda + \ldots + \lambda^{m-1} + 2 \cdot \lambda^m\\
     & \leq & 2 \cdot \lambda^{m+1}\\
     & = & O((\iota(m, \cl{\widetilde{C}}))^{m+1})
\end{array}
\]
showing the last statement of Theorem~\ref{thm:FV-decomposition}.

\section{Small models and elementary bounds}\label{section:elementary-bounds}

Following  is the central result of this section.

\begin{theorem}\label{thm:shrub-depth-DLS-and-index-bound}
Let $d, p, r \in \mathbb{N}$ be given. There exists an increasing function $h: \mathbb{N} \rightarrow \mathbb{N}$ such that the following are true for each $m \in \mathbb{N}_+$.  
\vspace{1pt}\begin{enumerate}
    \item  
    For every graph  $G \in \tm{r, p}{d}$, there exists $H \in \tm{r, p}{d}$ such that (i) $H \subseteq G$. (ii) $|H|$ is at most $\tower{d, h(d) \cdot m \cdot (m + \log r + \log p)}$, and (ii) $H \mequiv{m} G$. \label{thm:shrub-depth-DLS-and-index-bound:DLS} 
    \item 
    The index $\iota(m, \tm{r, p}{d}, \mso)$ of the $\mequiv{m}$ relation over $\tm{r, p}{d}$ is at most $\tower{d+1, h(d) \cdot m^2 \cdot (\log r + \log p)^2}$.\label{thm:shrub-depth-DLS-and-index-bound-for-m:index}
\end{enumerate}
\end{theorem}

We prove Theorem~\ref{thm:shrub-depth-DLS-and-index-bound} by first show the following core result for $\cl{T}_{d, p}$, and then transferring the latter to Theorem~\ref{thm:shrub-depth-DLS-and-index-bound} using $\fo$ interpretations. 
\begin{theorem}\label{thm:DLS-and-index-bound-for-m}
Let $d, p \in \mathbb{N}$ be given. There exists an increasing function $g: \mathbb{N} \rightarrow \mathbb{N}$ such that if $\zeta_{d, p}: \mathbb{N} \times \mathbb{N} \rightarrow \mathbb{N}$ is the function given by $\zeta_{d, p}(n_1, n_2) = \tower{n_2, g(d) \cdot (n_1 + 1) \cdot (n_1 + \log p)}$, then the following are true for each $m \in \mathbb{N}_+$.  
\vspace{1pt}\begin{enumerate}
    \item  
    For every tree $\tree{t} \in \cl{T}_{d, p}$, there exists a leaf-hereditary subtree $\tree{t}'$ of $\tree{t}$ such that (i) the heights of $\tree{t}'$ and $\tree{t}$ are the same, (ii) $|\tree{t}'|$ is at most $\zeta_{d, p}(m, d)$, and (ii) $\tree{t}' \mequiv{m} \tree{t}$.\label{thm:DLS-and-index-bound-for-m:DLS} 
    \item 
    The index $\iota(m, \cl{T}_{d, p}, \mso)$ of the $\mequiv{m}$ relation over $\cl{T}_{d, p}$ is at most $\zeta_{d, p}(m, d+1)$ if $d \ge 1$, and is $p$ if $d = 0$. \label{thm:DLS-and-index-bound-for-m:index}
\end{enumerate}
\end{theorem}

\begin{proof}[Proof of Theorem~\ref{thm:shrub-depth-DLS-and-index-bound}]
Since $G \in \tm{r, p}{d}$, from Section~\ref{section:prelims} there exists a tree model $(\tree{s}, S) \in \treemodel{r, p}{d}$ for $G$ and an $\fo$ interpretation $\Xi_{S, p}$ such that $G = \Xi_{S, p}(\tree{s})$. Let $q = \max\{\max\{\rank{\Xi_{S, p}} $ $ \mid S \subseteq [r]^2 \times [p]\}, \rank{\Omega_{r, p, d}}\}$ where $\Omega_{r, p, d}$ is as defined in Section~\ref{section:shrub-depth}; we know that $q = O(d)$. Since $\tree{s}$ can be seen as a tree in $\cl{T}_{d, r \cdot p + 1}$, by Theorem~\ref{thm:DLS-and-index-bound-for-m}(\ref{thm:DLS-and-index-bound-for-m:DLS}), there exists a leaf-hereditary subtree $\tree{s}'$ of $\tree{s}$ such that (i) $|\tree{s}'| \leq \zeta_{d, r \cdot p  + 1}(q + m, d)$, and (ii) $\tree{s}' \mequiv{q + m} \tree{s}$. Then $\tree{s}'$ models $\Omega_{r, p, d}$ and hence $(\tree{s}', S) \in \treemodel{r, p}{d}$ is a tree model for $H = \Xi_{S, p}(\tree{s}') \in \tm{r, p}{d}$. From the properties of the interpretation $\Xi_{S, p}$ (cf. Section~\ref{section:shrub-depth}), we infer the following: (i) Since $\tree{s}'$ is a leaf-hereditary subtree of $\tree{s}$, we have $H \subseteq G$; (ii) Since $V(H)$ is the set of leaves of $\tree{s}'$, we have $|H| \leq |\tree{s}'| \leq \zeta_{d, r \cdot p  + 1}(q + m, d) = \tower{d, g(d) \cdot (q + m + 1) \cdot (q + m + \log (r \cdot p + 1)} \leq  \tower{d, g(d) \cdot q^2 \cdot m \cdot (m + \log r + \log p)} \leq \tower{d, h(d) \cdot m \cdot (m + \log r + \log p)}$ where $h(d) = c_0 \cdot g(d) \cdot d^2$,  $q \leq c_1 \cdot d $ and $c_0 \ge (c_1)^2$; (iii) Since $\tree{s}' \mequiv{q+m} \tree{s}$, we have $H \mequiv{m} G$. 
\newcommand{\tms}[2]{\ensuremath{\mathrm{TM}^S_{#1}(#2)}}

We now look at the index of the $\mequiv{m}$ relation over $\tm{r, p}{d}$. For a given signature $S \subseteq [r]^2 \times [p]$, denote $\tms{r, p}{d}$ denote the subclass of $\tm{r, p}{d}$ of those graphs that have a tree model $(\tree{t}, S) \in \treemodel{r, p}{d}$, and $\mathrm{Mod}(\Omega_{r, p, d})$ denote the class of models of $\Omega_{r, p, d}$ in $\cl{T}_{d, r\cdot p +1}$. Then $\Xi_{S, p}$ is a surjective map from $\mathrm{Mod}(\Omega_{r, p, d})$ to $\tms{r, p}{d}$. For $q$ as above, since any equivalence class of the $\mequiv{q+m}$ relation over $\mathrm{Mod}(\Omega_{r, p, d})$ gets mapped by $\Xi_{S, p}$ to a subclass of an equivalence class of the $\mequiv{m}$ relation over $\tms{r, p}{d}$, we get by the surjectivity of $\Xi_S$ that $i(m, \tms{r, p}{d}, \mso) \leq i(q + m, \mathrm{Mod}(\Omega_{r, p, d}), \mso) \leq i(c_1 \cdot d + m, \cl{T}_{d, r\cdot p + 1}, \mso)$. Then $\iota(m, \tm{r, p}{d}, \mso) \leq \sum_{S \subseteq [r]^2 \times [p]} i(m, \tms{r, p}{d}, \mso) \leq 2^{r^2 \cdot p} \cdot i(c_1 \cdot d + m, \cl{T}_{d, r\cdot p + 1}, \mso) \leq 2^{r^2 \cdot p} \cdot \tower{d+1, g(d) \cdot (c_1 \cdot d + m +1) \cdot (c_1 \cdot d + m + \log (r \cdot p + 1))} \leq \tower{d+1, h(d) \cdot (m \cdot (\log r  + \log p))^2}$. 
\end{proof}

We now look at the proof of Theorem~\ref{thm:DLS-and-index-bound-for-m}. The proof uses the following lemma in a central way. The lemma can be verified using a simple {\ef} game argument.


\begin{lemma}\label{lemma:monadic-structures}
Let $\sigma$ be a finite vocabulary consisting of only monadic relation symbols. Let $\str{A}$ be an arbitrary (finite or infinite) $\sigma$-structure such that every element of $\str{A}$ is in the interpretation of exactly one predicate of $\sigma$. Let $a$ be a given element of $\str{A}$, and let $q \in \mathbf{N}_+$ and $q_1 = (q - 1) \cdot |\sigma|$. Then the following are true:
\vspace{1pt}\begin{enumerate}
    \item If $|\str{A}| > q_1$, then for every cardinal $\lambda$ such that $q_1 < \lambda \leq |\str{A}|$, there exists a substructure $\str{B}$ of $\str{A}$ such that (i) $\str{B}$ contains $a$, (ii) $|\str{B}| = \lambda$, and (iii) $\str{B} \fequiv{q} \str{A}$.\label{lemma:monadic-structures:1}
    \item If $|\str{A}| > q_1$, then for every cardinal $\lambda \ge |\str{A}|$, there exists a $\sigma$-structure $\str{B}$ containing $\str{A}$ as substructure such that (i) $|\str{B}| = \lambda$, and (ii) $\str{B} \fequiv{q} \str{A}$. Further, for any $T \in \sigma$ such that $|T^{\str{A}}| \ge q$, we can take $\str{B}$ to be a structure such that all elements of $\str{B}$ not in $\str{A}$ are in $T^{\str{B}}$, and none of these elements is in the interpretation of any other predicate of $\sigma$ in $\str{B}$.
    \label{lemma:monadic-structures:2}
\end{enumerate}
\end{lemma}

We will also need the following lemma. The proof is easy and is skipped.
\begin{lemma}\label{lem:tree-forest-transfer}
Let $\tree{t}_i$ for  $i = 1, 2$ be a rooted tree and let $\forest{f}_i$ be the forest of rooted trees obtained by removing the root of $\tree{t}_i$. Then $\tree{t}_1 \mequiv{m} \tree{t}_2$ if, and only if, $\forest{f}_1 \mequiv{m} \forest{f}_2$. 
\end{lemma}

\begin{proof}[Proof of Theorem~\ref{thm:DLS-and-index-bound-for-m}]
We prove the theorem by induction on $d$. The base case of $d = 0$ is trivial to see by taking $g(0) = 1$. Assume as induction hypothesis that the statement is true for $d - 1$ for $d \ge 1$. 

Consider a tree $\tree{t} \in \cl{T}_{d, p}$ of height equal to $d$. Let $\forest{f}$ be the forest of $p$-labeled rooted trees obtained by removing the root of $\tree{t}$. Specifically, let $\forest{f} = \bigcupdot_{i \in I} \tree{s}_i$ where $\tree{s}_i \in \cl{T}_{d-1, p}$ and the (unrooted tree underlying) $\tree{s}_i$ is a child subtree of $\tree{t}$, and $I$ is a possibly infinite set. Let $\cl{S}$ be the class of rooted forests whose constituent trees belong to $\cl{T}_{d-1, p}$; so $\forest{f} \in \cl{S}$. Consider now the $\mso[m]$-type indicator $\ind{m}{\mc{F}}$ for the family $\mc{F} = (\tree{s}_i)_{i \in I}$ and the sentence $\Phi := \Theta_\delta$ for $\delta = \mclass{m}{\cl{S}}(\forest{f})$ (the $\mso[m]$ sentence axiomatizing the $\mequiv{m}$ equivalence class of $\forest{f}$ in $\cl{S}$).
By Theorem~\ref{thm:FV-decomposition}, there exists an FO sentence $\alpha_\Phi$ over the vocabulary $\tau_{m, \cl{T}_{d-1, p}}$ such that 
\[
\ind{m}{\mc{F}} \models \alpha_\Phi~~\mbox{if, and only if,}~~\forest{f} \models \Phi
\]

Now we know from Theorem~\ref{thm:FV-decomposition} that for any $\mso[m]$ sentence $\Psi$ over the vocabulary of $\cl{T}_{d-1, p}$, the sentence $\alpha_\Psi$ given by the theorem has rank that is $O((\iota(m, \widetilde{\cl{T}}_{d-1, p}, \mso))^{m+1})$ where $\widetilde{\cl{T}}_{d-1, p}$ is the expansion of $\cl{T}_{d-1, p}$ with $m$ new unary predicates. Now there is a natural 1-1 correspondence between $\widetilde{\cl{T}}_{d-1, p}$ and $\cl{T}_{d-1, p \cdot 2^m}$, and two $\widetilde{\cl{T}}_{d-1, p}$ structures are $\mso[m]$-equivalent iff their corresponding $\cl{T}_{d-1, p \cdot 2^m}$ structures are. Then the rank of $\alpha_\Psi$ is $O((\iota(m, \cl{T}_{d-1, p \cdot 2^m}, \mso))^{m+1})$. Let $c_0$ be a constant such that $\mbox{rank}(\alpha_\Psi) \leq c_0 \cdot (\iota(m, \cl{T}_{d-1, p \cdot 2^m}, \mso))^{m+1}$ for all $m \ge m_0$. In fact, as shown in Section~\ref{section:FV-comp}, the values of $c_0$ and $m_0$ can be taken to be $c_0 = 2$ and $m_0 = 1$.

Returning to $\alpha_\Phi$, we get that $\mbox{rank}(\alpha_\Phi) \leq c_0 \cdot (\iota(m, \cl{T}_{d-1, p \cdot 2^m}, \mso))^{m+1}$. Now by induction hypothesis, we have that (i) if $d - 1 = 0$, then $\iota(m, \cl{T}_{d-1, p \cdot 2^m}, \mso) =  p \cdot 2^m$, and (ii) if $d > 1$, then  $\iota(m, \cl{T}_{d-1, p \cdot 2^m}, \mso) \leq \zeta_{d-1, p \cdot 2^m}(m, d) = \tower{d, g(d-1) \cdot (m + 1) \cdot (2m + \log p))} \leq \tower{d, 2 \cdot g(d-1) \cdot (m + 1) \cdot (m + \log p)}$. Then the rank of $\alpha_\Phi$ is at most 
\vspace{1pt}\begin{enumerate}
    \item $c_0 \cdot (p \cdot 2^m)^{m+1} \leq  2^{c_0 \cdot (m + 1) \cdot (m + \log p))}$  if $d = 1$, and 
    \item $c_0 \cdot \tower{d,  2 \cdot g(d-1) \cdot (m + 1) \cdot (m + \log p)})^{m+1} \leq \tower{d, c_0 \cdot 4 \cdot g(d-1) \cdot (m + 1) \cdot (m + \log p)}$, if $d > 1$.
\end{enumerate} 

If $\rho_{d, p}: \mathbb{N} \rightarrow \mathbb{N}$ is the function given by $\rho_{d, p}(m) = \tower{d+1, 4 \cdot c_0 \cdot g(d) \cdot (m + 1) \cdot (m + \log p)}$, then we see that in either case above, the rank of $\alpha_\Phi$ is at most $\rho_{d-1, p}(m)$. We also see that $\tau_{m, \cl{T}_{d-1, p}}$ is the vocabulary which contains one unary predicate symbol for every element of $\alltypes{m, \cl{T}_{d-1, p}}$, and only those predicates; then $|\tau_{m, \cl{T}_{d-1, p}}| = \iota(m, \cl{T}_{d-1, p}, \mso)$ which is equal to $p$ if $d - 1 = 0$, and at most $\tower{d, g(d-1) \cdot (m + 1) \cdot (m + \log p))}$ if $d > 1$. Then $\mbox{rank}(\alpha_\Phi) \cdot |\tau_{m, \cl{T}_{d-1, p}}| \leq \rho_{d-1, p}(m) \cdot |\tau_{m, \cl{T}_{d-1, p}}| \leq \tower{d, 5 \cdot c_0 \cdot g(d-1) \cdot (m + 1) \cdot (m + \log p)}$ for all $d \ge 1$.

We observe now that $\ind{m}{\mc{F}}$ is a structure over a finite monadic vocabulary such that each element of its universe is in the interpretation of exactly one predicate in the vocabulary. Let $\tree{s}_{i^*} \in \mc{F}$ for $i^* \in I$ be such that the height of $\tree{s}_{i^*}$ is equal to $d-1$ (there must be such a tree in $\mc{F}$ since height of $\tree{t}$ is equal to $d$). Recall that $I$ is the universe of $\ind{m}{\mc{F}}$. Then by Lemma~\ref{lemma:monadic-structures}(\ref{lemma:monadic-structures:1}), taking $a = i^*$ and $q = \rho_{d-1, p}(m)$ in the lemma, we get that there exists a substructure $\str{B}$ of $\ind{m}{\mc{F}}$ such that (i) $\str{B}$ contains $i^*$, (ii) $|\str{B}| \leq 1 + (q - 1) \cdot |\tau_{m, \cl{T}_{d-1, p}}| \leq q \cdot |\tau_{m, \cl{T}_{d-1, p}}| \leq \tower{d, 5 \cdot c_0 \cdot g(d-1) \cdot (m + 1) \cdot (m + \log p)}$, and (iii) $\str{B} \fequiv{q} \ind{m}{\mc{F}}$. Then $\str{B}$ can be seen as the $\mso[m]$-type indicator $\ind{m}{\mc{F}'}$ of the family $\mc{F}' = (\tree{s}_j)_{j \in I'}$. for a subset $I' \subseteq I$, that contains $i^*$. Then by Theorem~\ref{thm:FV-decomposition}, we have
\[
\ind{m}{\mc{F}'} \models \alpha_\Phi~~\mbox{if, and only if,}~~\bigcupdot_{j \in I'} \tree{s}_j \models \Phi
\]
Since (i) $\ind{m}{\mc{F}'} = \str{B} \fequiv{q} \ind{m}{\mc{F}}$, (ii) $\ind{m}{\mc{F}} \models \alpha_\Phi$, and (iii) $\mbox{rank}(\alpha_\Phi) \leq q$, we have $\ind{m}{\mc{F}'} \models \alpha_\Phi$ and therefore $\bigcupdot_{j \in I'} \tree{s}_j \models \Phi$. Then $\bigcupdot_{j \in I'} \tree{s}_j \mequiv{m} \forest{f}$. 

We now utilize the induction hypothesis for Part (\ref{thm:DLS-and-index-bound-for-m:DLS}). Since $\tree{s}_j \in \cl{T}_{d-1, p}$ for all $j \in I'$, there exists a leaf-hereditary subtree $\tree{s}'_j$ of $\tree{s}_j$ such that (i) the heights of $\tree{s}'_j$ and $\tree{s}_j$ are the same, (ii) $|\tree{s}'_j| \leq \zeta_{d-1, p}(m, d-1)$, and (ii) $\tree{s}'_j \mequiv{m} \tree{s}_j$. Then consider the forest $\forest{f}' = \bigcupdot_{j \in I'} \tree{s}'_j$. Since disjoint union satisfies the Feferman-Vaught composition property (see Section~\ref{section:prelims}), we get that $\forest{f}' \mequiv{m} \bigcupdot_{j \in I'} \tree{s}_j \mequiv{m} \forest{f}$. Further, we have $1 + |\forest{f}'| = 1 + \sum_{j \in I'} |\tree{s}'_j| \leq  1 + |I'| \cdot \max_{j \in I'} |\tree{s}'_j| \leq 1 +  \tower{d, 5 \cdot c_0 \cdot g(d-1) \cdot (m + 1) \cdot (m + \log p)} \cdot \tower{d-1, g(d-1) \cdot (m + 1) \cdot (m + \log p)} \leq \tower{d, 6 \cdot c_0 \cdot g(d-1) \cdot (m + 1) \cdot (m + \log p)}$ for all $d \ge 1$. Let $\tree{t}'$ be the leaf-hereditary subtree of $\tree{t}$ obtained by removing all child subtrees (that are the unrooted versions of) $\tree{s}_j$ for $j \notin I'$, and replacing $\tree{s}_j$ with $\tree{s}'_j$ (again the replacement being for the unrooted versions of the trees) for each $j \in I'$ in $\tree{t}$. Then $\forest{f'}$ is indeed the forest of rooted trees obtained by deleting the root of $\tree{t}'$. Observe that $\forest{f}'$ contains the tree $\tree{s}'_{i^*}$ (since $i^* \in I'$) whose height is the same as that of $\tree{s}_{i^*}$ whose height in turn is equal to $d-1$; then $\tree{t}'$ has height equal to $d$ which is the height of $\tree{t}$. Further, since $\forest{f}' \mequiv{m} \forest{f}$, we get by Lemma~\ref{lem:tree-forest-transfer}, that $\tree{t}' \mequiv{m} \tree{t}$. 
Finally, $|\tree{t}'| = |\forest{f}'|+ 1 \leq \tower{d, 6 \cdot c_0 \cdot g(d-1) \cdot (m + 1) \cdot (m + \log p)}$ for all $d \ge 1$. 

We now observe that the existence of $\tree{t}'$ as above for every tree $\tree{t}$ in $\cl{T}_{d, p}$ implies that $\iota(m, \cl{T}_{d, p}, \mso)$ is at most the number of structures of $\cl{T}_{d, p}$ whose size is at most $\tower{d, 6 \cdot c_0 \cdot g(d-1) \cdot (m + 1) \cdot (m + \log p)}$. Since the number of structures of $\cl{T}_{d, p}$ with universe size at most $\mu$ for any number $\mu$ is at most $\mu \cdot 2^{\mu \cdot (\mu + \log p)} \leq 2^{3 \mu^2}$ if $\log p \leq \mu$, we get, by taking $\mu = \tower{d, 6 \cdot c_0 \cdot g(d-1) \cdot (m + 1) \cdot (m + \log p)}$, that $\iota(m, \cl{T}_{d, p}, \mso) \leq \tower{d+1, 14 \cdot c_0 \cdot g(d-1) \cdot (m + 1) \cdot (m + \log p)}$. Then defining $g(d) = 14 \cdot c_0 \cdot g(d-1)$, we see that both parts of the present theorem are true for $d$. This completes the induction and hence the proof.
\end{proof}

\begin{remark}\label{remark:alpha-phi-properties}
By the same reasoning as in the proof of Theorem~\ref{thm:DLS-and-index-bound-for-m}, it follows that if $\cl{C}$ in Theorem~\ref{thm:FV-decomposition} is taken to be $\cl{T}_{d, p}$, then for the $\mso[m]$ sentence $\Phi$ as considered in Theorem~\ref{thm:FV-decomposition}, so an arbitrary $\mso$ sentence of rank $m$ over the vocabulary of $\cl{T}_{d, p}$, the sentence $\alpha_\Phi$ is such that (i) the vocabulary $\tau_{m, \cl{T}_{d, p}}$ of $\alpha_\Phi$ has size at most $\zeta_{d, p}(m, d+1)$, (ii) the rank of $\alpha_\Phi$ is at most $\rho_{d, p}(m)$, and (iii) $\mbox{rank}(\alpha_\Phi) \cdot |\tau_{m, \cl{T}_{d, p}}| \leq \zeta_{d+1, p}(m, d+1)$.
\end{remark}

The proof of Theorem~\ref{thm:DLS-and-index-bound-for-m} actually shows a stronger relationship between $\tree{t}'$ and $\tree{t}$. We recall the function $\rho_{d, p}: \mathbb{N} \rightarrow \mathbb{N}$ introduced in the proof of Theorem~\ref{thm:DLS-and-index-bound-for-m}, defined as $\rho_{d, p}(m) = \tower{d+1, 4 \cdot g(d) \cdot (m  + 1) \cdot (m + \log p)}$ where $g$ is the function given by Theorem~\ref{thm:DLS-and-index-bound-for-m}. We now have the following definition.

\begin{definition}\label{defn:chaineq}
Let $d, p, m \in \mathbb{N}$ be given. For $\tree{t}_1, \tree{t}_2 \in \cl{T}_{d, p}$ both of the same height, the binary relation $\tree{t}_1 \chaineq_m \tree{t}_2$ is defined inductively on the height $h \in \{0, \ldots, d\}$ of the trees as follows: 
\vspace{1pt}\begin{enumerate}
    \item If $h = 0$, then $\tree{t}_1 \chaineq_{m} \tree{t}_2$ if $\tree{t}_1 \cong \tree{t}_2$.
    \item Assume $\chaineq_{m}$ has been defined when the height of $\tree{t}_1$ and $\tree{t}_2$ is $h \leq k < d$. Suppose now that $\tree{t}_1$ and $\tree{t}_2$ have height equal to $k+1$. For $i \in \{1, 2\}$, let $\forest{f}_i$ be the forest of rooted trees of $\cl{T}_{d-1, p}$ obtained by removing the root of $\tree{t}_i$ and let $\mc{F}_i$ be the family of trees constituting $\forest{f}_i$. Let $\str{A}_i = \ind{m}{\mc{F}_i}$ be the $\mso[m]$-type indicator for $\mc{F}_i$. 
    
    We now say that $\tree{t}_1 \chaineq_{m} \tree{t}_2$ holds if: (i) the roots of $\tree{t}_1$ and $\tree{t}_2$ are the same, and have the same labels;\label{defn:chaineq:cond-1}  (ii)  $\str{A}_1 \fequiv{q} \str{A}_2$ where $q = \rho_{k, p}(m)$; and (iii) for each tree $\tree{t}_1' \in \mc{F}_1$, there exists a tree $\tree{t}_2' \in \mc{F}_2$ such that $\tree{t}_1' \chaineq_{m} \tree{t}_2'$.
\end{enumerate} 
\end{definition} 
\vspace{1pt}The tree $\tree{t}_2'$ in the above definition is clearly unique for $\tree{t}_1'$. We now observe the following properties of the $\chaineq_m$ relation, the first two of which can be verified by an easy induction, and the last of which can be verified by a very similar reasoning as in the proof of Theorem~\ref{thm:DLS-and-index-bound-for-m}.

\begin{lemma}\label{lem:chaineq-props}
For $d, p, m \in \mathbb{N}$, the following are true of the $\chaineq_m$ relation on $\cl{T}_{d, p}$:
\vspace{1pt}\begin{enumerate}
    \item $\chaineq_{m}$ is reflexive and transitive.
    \item Condition (\ref{defn:chaineq:cond-1}.i) in the $\chaineq_m$ definition can be replaced with `` $\tree{t}_1$ is a leaf-hereditary subtree of $\tree{t}_2$" to get an equivalent definition.
    \item  For trees $\tree{t}_1, \tree{t}_2 \in \cl{T}_{d, p}$, if $\tree{t}_1 \chaineq_m \tree{t}_2$, then (i) $\tree{t}_1 \chaineq_{m'} \tree{t}_2$ for all $m' \leq m$, and (ii) $\tree{t}_1 \mequiv{m} \tree{t}_2$. 
\end{enumerate}
\end{lemma}

We now obtain the following by an entirely analogous set of arguments as for Theorem~\ref{thm:DLS-and-index-bound-for-m}.

\begin{corollary}\label{cor:DLS-chaineq-form}
 For every $d, p, m \in \mathbb{N}$ and every tree $\tree{t} \in \cl{T}_{d, p}$, there exists a tree $\tree{t}' \in \cl{T}_{d, p}$ such that (i) $\tree{t}' \chaineq_{m} \tree{t}$, and (ii) $|\tree{t}'| \leq \zeta_{d, p}(m, d)$ where $\zeta_{d, p}$ is the function given by Theorem~\ref{thm:DLS-and-index-bound-for-m}.
\end{corollary}

Before we proceed, we define the following functions $\chi_{d, p}, \xi_{d, p}: \mathbb{N} \rightarrow \mathbb{N}$ for $d \ge 1$. Observe that both of the functions are strictly increasing.
\begin{equation}\label{eqn:quantities}
  \begin{aligned}
    \chi_{d, p}(m) & = \tower{d+1, g(d+1) \cdot (m + 1) \cdot (m + \log p)} ~~~( =  \zeta_{d+1, p}(m, d+1) )\\
    \xi_{d, p}(m) & =  \tower{d+1, g(d+1) \cdot (d + 1) \cdot (m + 1) \cdot (m + \log p)}\\
  \end{aligned}
\end{equation}

An important consequence of Theorem~\ref{thm:FV-decomposition} and Theorem~\ref{thm:DLS-and-index-bound-for-m} is the following lemma which will be of much use to us in the next section. 

\begin{lemma}\label{lem:degree-regulation}
Let $d, p, m \in \mathbb{N}$ be given. Let $\mc{F}  = (\tree{s}_i)_{i \in I}$ be a family of trees of $\cl{T}_{d, p}$ where $I$ is an index set of arbitrary cardinality. Let $\str{A} = \ind{m}{\mc{F}}$ be the $\mso[m]$ type indicator of $\mc{F}$ and let $i^* \in I$ be a given element of $\str{A}$. Then the following are true for any $q \in \mathbb{N}$:
\vspace{1pt}\begin{enumerate}
    \item If $|I| > \chi_{d, p}(m)$, then for every cardinal $\lambda$ such that $\chi_{d, p}(m) \leq \lambda \leq |I|$, there exists a subset $I' \subseteq I$ such that (i) $I'$ contains $i^*$,  (ii) $|I'| = \lambda$, and (iii) if $\mc{F}' = (\tree{s}_i)_{i \in I'}$ and $\str{B} = \ind{m}{\mc{F}'}$, then $\str{B} \fequiv{q} \str{A}$ where $q = \rho_{d, p}(m)$.\label{lem:degree-regulation-downward}
    \item If $|I| > \chi_{d, p}(m)$, then for every cardinal $\lambda$ such that $\lambda \ge |I|$, there exists a set $I'$ of which $I$ is a subset and trees $\tree{s}_i \in \cl{T}_{d, p}$ for $i \in I' \setminus I$ such that (i) $|I'| = \lambda$, and (ii) if $\mc{F}' = (\tree{s}_i)_{i \in I'}$ and $\str{B} = \ind{m}{\mc{F}'}$, then $\str{B} \fequiv{q} \str{A}$ where $q = \rho_{d, p}(m)$. Further, there exists $j \in I$ such that all the trees $\tree{s}_i$ for $i \in I' \setminus I$ can be taken to be isomorphic to $\tree{s}_j$.\label{lem:degree-regulation-upward}
\end{enumerate}
\end{lemma}

\begin{proof}
The proof is very similar to the proof of Part~(\ref{thm:DLS-and-index-bound-for-m:DLS}) of  Theorem~\ref{thm:DLS-and-index-bound-for-m}.
Let $\forest{f}$ be the forest given by $\forest{f} = \bigcupdot_{i \in I} \tree{s}_i$. As in the proof of Theorem~\ref{thm:DLS-and-index-bound-for-m}, let $\cl{S}$ be the class of rooted forests whose constituent trees are all in $\cl{T}_{d, p}$; then $\forest{f} \in \cl{S}$. Consider the sentence $\Phi := \Theta_\delta$ for $\delta = \mclass{m}{\cl{S}}(\forest{f})$.
By Theorem~\ref{thm:FV-decomposition}, there exists an $\fo$ sentence $\alpha_\Phi$ of over the vocabulary $\tau_{m, \cl{T}_{d, p}}$ such that
\[
\str{A} \models \alpha_\Phi~~\mbox{if, and only if,}~~\forest{f} \models \Phi
\]
By Remark~\ref{remark:alpha-phi-properties}, we have that $\mbox{rank}(\alpha_\Phi) \leq \rho_{d, p}(m)$ and that $\mbox{rank}(\alpha_\Phi) \cdot |\tau_{m, \cl{T}_{d, p}}| \leq \chi_{d, p}(m)$.
Once again, we observe that $\str{A}$ is a structure over a finite monadic vocabulary such that each element of its universe is in the interpretation of exactly one predicate in the vocabulary. Then by Lemma~\ref{lemma:monadic-structures}, taking $q = \rho_{d, p}(m)$, we get that for every $\lambda \neq |I|$ and $\lambda \ge \chi_{d, p}(m) > (\mbox{rank}(\alpha_\Phi) - 1) \cdot |\tau_{m, \cl{T}_{d, p}}|$, there exists a structure $\str{B}$ such that (i) $\str{B}$ is a substructure of $\str{A}$ containing $i^*$ if $\lambda < |I|$, and $\str{A}$ is a substructure of $\str{B}$ if $\lambda > |I|$, (ii) $|\str{B}| = \lambda$, and (ii) $\str{B} \fequiv{q} \str{A}$. If $I'$ is the universe of $\str{B}$, then $\str{B}$ can be seen as the $\mso[m]$-type indicator of the family $\mc{F}' = (\tree{s}_i)_{i \in I'}$ where if $I' \setminus I \neq \emptyset$ (when $\lambda > |I|$), then $\tree{s}_i$ for $i \in I' \setminus I$ is a tree in $\cl{T}_{d, p}$. We see that $I'$ always contains $i^*$, the size of $I'$ is $\lambda$, and as shown above $\str{B} \fequiv{q} \str{A}$.

To see the last statement of the lemma, we see by Lemma~\ref{lemma:monadic-structures}(\ref{lemma:monadic-structures:2}) that for any $T \in \tau_{m, \cl{T}_{d, p}}$ such that $|T^{\str{A}}| \ge q$, all the elements of $I' \setminus I$ above can be taken such that they are all in the interpretation of $T$ in $\str{B}$ and in the interpretation of no other predicate of $\tau_{m, \cl{T}_{d, p}}$ in $\str{B}$. In other words, one can choose a tree $\tree{s}_j$ for $j \in I$ such that $j \in T^{\str{A}}$ (so $T = \mclass{m}{\cl{S}}(\tree{s}_j)$), and for all $i \in I' \setminus I$, it holds that $i \in T^{\str{B}}$ (so $T = \mclass{m}{\cl{S}}(\tree{s}_i)$). In other words, $\tree{s}_i \mequiv{m} \tree{s}_j$ for all $i \in I' \setminus I$. In particular then, $\tree{s}_i$ can be taken to be isomorphic to $\tree{s}_j$ for all $i \in I' \setminus I$.
\end{proof}

\section{The extended L\"owenheim-Skolem property for $\tm{r, p}{d}$}\label{section:ELS}

We prove Theorem~\ref{thm:main-display} in this section. We prove the theorem in two parts, first for the case when one of $\eta$ or  $\lambda$ in the $\melsref$ property described in Section~\ref{section:intro} is finite, and the other when both of these cardinals are infinite. We recall that a function $f: \mathbb{N}_+ \rightarrow \mathbb{N}$ is a \emph{scale function} if it is strictly increasing. For a cardinal $\lambda$, the $\lambda^{th}$ scale, denoted $\scale{\lambda}{f}$, is defined as follows. If $\lambda$ is finite, then $\scale{\lambda}{f}$ is the interval $[f(\lambda) + 1, f(\lambda+1)] = \{j \mid f(\lambda) + 1 \leq j \leq f(\lambda+1)\}$ for $\lambda > 0$, and $[1, f(1)] = \{j \mid 1 \leq  j \leq f(1)\}$ for $\lambda = 0$. If $\lambda$ is infinite, then $\scale{\lambda}{f} = \{\lambda\}$. 

\begin{theorem}\label{thm:shrub-depth-full-LS-finite}
Let $d, p \in \mathbb{N}$ be given. There exists a $(d+1)$-fold exponential scale function $\Upsilon_{r, p, d}: \mathbb{N}_+ \rightarrow \mathbb{N}$ such that the following is true. Let $G \in \tm{r, p}{d}$ be such that $|G| \in \scale{\eta}{\vartheta_{r, p, d}}$ for some (possibly infinite) $\eta \ge 1$. 
\vspace{1pt}\begin{enumerate}
    \item (Downward $\mso$-$\els$) For all cardinals $\lambda \leq \eta$, if $\lambda$ is finite, then there exists $H \in \tm{r, p}{d}$ such that (a) $H \subseteq G$, (b) $|H| \in \scale{\lambda}{\Upsilon_{r, p, d}}$, and (c) $H \mequiv{\lambda} G$.\label{thm:shrub-depth-full-LS-finite-downward}
    \item (Upward $\mso$-$\els$) For all cardinals $\lambda \geq \eta$, if $\eta$ is finite, then there exists $H \in \tm{r, p}{d}$ such that (a) $G \subseteq H$, (b) $|H| \in \scale{\lambda}{\Upsilon_{r, p, d}}$, and (c) $H \mequiv{\eta} G$.\label{thm:shrub-depth-full-LS-finite-upward}
\end{enumerate} 
\end{theorem}

\begin{prop}\label{prop:shrub-depth-full-LS-infinite}
Let $d, p \in \mathbb{N}$ be given. Let $G \in \tm{r, p}{d}$ be such that $|G| = \eta$ for an infinite cardinal $\eta$. Then for any infinite cardinal $\lambda$, there exists $H \in \tm{r, p}{d}$ such that (i) $H \subseteq G$ if $\lambda \leq \eta$, and $G \subseteq H$ if $\eta \leq \lambda$, (ii) $|H| = \lambda$, and (iii) $H \mequiv{m} G$ for all $m \in \mathbb{N}$. 
\end{prop}

\begin{proof}[Proof of Proposition~\ref{prop:shrub-depth-full-LS-infinite}]
It is known~\cite{CF20} that there exists an $\fo$ sentence $\varphi_1$ that defines $\tm{r}{d}$ over all structures. The property of $p$-labeling -- that the unary predicates $P_1. \ldots, P_p$ are mutually exclusive and exhaustive  (and allowed to be empty) -- is easily defined by an $\fo$ sentence $\varphi_2$. Then the sentence $\varphi_1 \wedge \varphi_2$ defines $\tm{r, p}{d}$ over all structures. Now using the fact that $\mso$ collapses to $\fo$ over $\tm{r, p}{d}$ as shown by Theorem~\ref{thm:mso=fo} in Section~\ref{section:mso=fo}, the statement of the present proposition is just the classical {\lsref} theorem which is known to be true over elementary classes of structures. 
\end{proof}

The proof of Theorem~\ref{thm:shrub-depth-full-LS-finite} proceeds by first showing Theorem~\ref{thm:full-LS-finite} below which is the analogue of Theorem~\ref{thm:shrub-depth-full-LS-finite} for $\cl{T}_{d, p}$. This is the core result which can then be easily transferred to $\tm{r, p}{d}$ via $\fo$ interpretations. 

\begin{theorem}\label{thm:full-LS-finite}
Let $d, p \in \mathbb{N}$ be given. There exists a $(d+1)$-fold exponential scale function $\vartheta_{d, p}: \mathbb{N}_+ \rightarrow \mathbb{N}$ such that the following is true. Let $\tree{t} \in \cl{T}_{d, p}$ be a given tree such that $|\tree{t}| \in \scale{\eta}{\vartheta_{d, p}}$ for some (possibly infinite) $\eta \ge 1$. 
\vspace{1pt}\begin{enumerate}
    \item (Downward) For all cardinals $\lambda \leq \eta$, if $\lambda$ is finite, then there exists a tree $\tree{t}' \in \cl{T}_{d, p}$ such that (i) $|\tree{t}'| \in \scale{\lambda}{\vartheta_{d, p}}$, and (ii) $\tree{t}' \chaineq_{\lambda} \tree{t}$.\label{thm:full-LS-finite-downward}
    \item (Upward) For all cardinals $\lambda \ge \eta$, if $\eta$ is finite, then there exists a tree $\tree{t}' \in \cl{T}_{d, p}$ such that (i) $|\tree{t}'| \in \scale{\lambda}{\vartheta_{d, p}}$, and (ii) $\tree{t} \chaineq_{\eta} \tree{t}'$.\label{thm:full-LS-finite-upward}
\end{enumerate} 
\end{theorem}

\begin{proof}[Proof of Theorem~\ref{thm:shrub-depth-full-LS-finite}]
Let $\vartheta_{d, p}$ be the $(d+1)$-fold exponential function given by Theorem~\ref{thm:full-LS-finite}.
Let $q_0 \in \mathbb{N}$ be such that $q_0 = \max\{\max\{\rank{\Xi_{S, p}} $ $ \mid S \subseteq [r]^2 \times [p]\}, \rank{\Omega_{r, p, d}}\}$ where $\Omega_{r, p, d}$ is as defined in Section~\ref{section:prelims}. Define $\Upsilon_{r, p, d}$ inductively as follows. Firstly, $\Upsilon_{r, p, d}(1) = \vartheta_{d, p}(q_0+1)$. Suppose $\Upsilon_{r, p, d}(i)$ has been defined for $i \ge 1$. Let $j$ be the smallest number such that $\vartheta_{d, p}(j) \ge (d+1) \cdot \Upsilon_{r, p, d}(i)$. Then $\Upsilon_{r, p, d}(i+1) = \vartheta_{d, p}(j+1)$. We verify that $\Upsilon_{r, p, d}$ is increasing and $(d+1)$-fold exponential, and that $\Upsilon_{r, p, d}(i) \ge \vartheta_{d, p}(i)$ for all $i \ge 1$. We now show that $\Upsilon_{r, p, d}$ is as desired. Let $G \in \tm{r, p}{d}$ be such that $|G| \in \scale{\eta}{\Upsilon_{r, p, d}}$. Let $(\tree{s}, S) \in \treemodel{r, p}{d}$ be a tree model for $G$; then $G = \Xi_{S, p}(\tree{s})$. Since the number of leaves of $\tree{s}$ equals $|G|$, and the height of $\tree{s}$ is $d$, we see that $|G| \leq |\tree{s}| \leq d \cdot |G| + 1 \leq (d+1) \cdot |G|$. Let $\eta'$ be such that $|\tree{s}| \in \scale{\eta'}{\vartheta_{d, p}}$. We now show Part~(\ref{thm:shrub-depth-full-LS-finite-downward}) of the theorem. Part~(\ref{thm:shrub-depth-full-LS-finite-upward}) can be done similarly.

Let $\lambda \leq \eta$ be finite. Assume $\lambda < \eta$ since if $\lambda = \eta$, we are done by taking $H = G$. Let $\lambda'$ be the smallest number such that $\vartheta_{d, p}(\lambda') \ge (d+1) \cdot \Upsilon_{r, d, p}(\lambda)$. Then $\lambda'$ is finite, and $\lambda' - q_0 \ge \lambda$ because $\Upsilon_{r, p, d}(i) \ge \vartheta_{d, p}(i)$, and $\vartheta_{d, p}$ and $\Upsilon_{r, d, p}$ are both increasing. Further $\lambda' < \eta'$ -- this is because $|\tree{s}| \ge |G| \ge \Upsilon_{r, p, d}(\eta) \ge  \Upsilon_{r, d, p}(\lambda+1) = \vartheta_{d, p}(\lambda' + 1) $
the last of these equalities by definition of $\Upsilon_{r, d, p}$. Then by Theorem~\ref{thm:full-LS-finite} and Lemma~\ref{lem:chaineq-props}, there exists a leaf-hereditary subtree $\tree{s}'$ of $\tree{s}$ such that (i) $|\tree{s}'| \in \scale{\lambda'}{\vartheta_{d, p}}$ and (ii) $\tree{s}' \mequiv{\lambda'} \tree{s}$. Now since $\vartheta_{d, p}(\lambda') \ge \Upsilon_{r, d, p}(\lambda) \ge \Upsilon_{r, d, p}(1) = \vartheta_{d, p}(q_0 + 1)$, we have $\lambda' \ge q_0 + 1$. Then $\tree{s}'$ models $\Omega_{r, d, p}$ so that $(\tree{s}', S) \in \treemodel{r, p}{d}$. Let $H = \Xi_S(\tree{s}')$. Then  since $\tree{s}'$ is a leaf-hereditary subtree of $\tree{s}$, we have $H \subseteq G$. Further since $ \lambda' \ge q_0 + \lambda$ and $\tree{s}' \mequiv{\lambda'} \tree{s}$, we get  that $H \mequiv{\lambda} G$  (cf. Section~\ref{section:shrub-depth}). Finally, we observe that $|H| \leq \tree{s}' \leq (d+1) \cdot |H|$. Since $|\tree{s}'| \leq \vartheta_{d, p}(\lambda'+1) = \Upsilon_{r, d, p}(\lambda+1)$ and since $(d+1) \cdot \Upsilon_{r, d, p}(\lambda) \leq \vartheta_{d, p}(\lambda') < |\tree{s}'|$, we get that indeed $|H| \in \scale{\lambda}{\Upsilon_{r, p, d}}$.
\end{proof}

We prove Theorem~\ref{thm:full-LS-finite} in the remainder of this section. We claim that the function $\vartheta_{d, p}$ in the statement of Theorem~\ref{thm:full-LS-finite} can be taken to be the function $\vartheta_{d, p}(\lambda) = \xi_{d, p}(\lambda)$ for all finite cardinals $\lambda \ge 1$. We prove this below. A key property available to us with this choice is that 
\begin{equation}\label{ineq:vartheta}
  \begin{aligned}
  \vartheta_{d, p}(\lambda) &\ge \vartheta_{d, p}(\lambda-1) + \xi_{d-1, p}(\lambda - 1) 
  \end{aligned}
\end{equation}
This is because $\xi_{d-1, p}(\lambda - 1) + \xi_{d, p}(\lambda-1) \leq 2 \cdot \xi_{d, p}(\lambda - 1)$ and the latter now is easily seen to be $\leq \xi_{d, p}(\lambda)$. We now prove the two parts of Theorem~\ref{thm:full-LS-finite} separately below. We will need the following 
two lemmas.
\begin{lemma}\label{lemma:simple-FV-comp}
Let $d, p, m \in \mathbb{N}$ be given. Let $\tree{t} \in \cl{T}_{d, p}$ be a given tree and $z$ be a node of $\tree{t}$ at height $h$.  Let $\tree{t}'$ be the tree obtained by replacing the (unrooted version of the) subtree $\tree{t}_z$ of $\tree{t}$ with (the unrooted version of) a tree $\tree{s} \in \cl{T}_{h, p}$. If $\tree{s} \chaineq_{m} \tree{t}_z$, then $\tree{t}' \chaineq_{m} \tree{t}$.
\end{lemma}

\begin{proof}
We show that the statement is true for the special case when $z$ is a child of $\root{\tree{t}}$. The general statement follows by repeated applications of this special case to the subtrees rooted at the nodes appearing along the path from the parent of $z$ to the root of $\tree{t}$, and the corresponding subtrees rooted at the nodes appearing along the path from the parent of $\root{\tree{s}}$ to the root of $\tree{t}'$ (here we are seeing (the unrooted version of)  $\tree{s}$ as a subtree of $\tree{t}'$).

Consider $\tree{t}$; let $\forest{f}$ be the forest of rooted trees obtained by deleting the root of $\tree{t}$; one of these rooted trees is $\tree{t}_z$. Let $\forest{f}'$ be the forest obtained by replacing $\tree{t}_z$ with $\tree{s}$, and keeping the other trees of $\forest{f}$ intact. Let $\mc{F}$ and $\mc{F}'$ resp. be the families of the trees in the forests $\forest{f}$ and $\forest{f}'$, and let $\str{A}$ and $\str{A}'$ be the $\mso[m]$-type indicators for $\mc{F}$ and $\mc{F}'$. Since $\tree{s} \chaineq_m \tree{t}_z$, it follows that $\tree{s} \mequiv{m} \tree{t}_z$, so that $\str{A} \cong \str{A}'$ and hence $\str{A}' \fequiv{q} \str{A}$ for any $q$. Also for every tree $\tree{x} \in \mc{F}'$, choose $\tree{y} \in \mc{F}$ as: $\tree{y} = \tree{x}$ if $\tree{x} \neq \tree{s}$ and $\tree{y} = \tree{t}_z$ otherwise; then we see that $\tree{x} \chaineq_m \tree{y}$ (since $\chaineq_m$ is reflexive; see Lemma~\ref{lem:chaineq-props}). Further the heights of $\tree{x}$ and $\tree{y}$ are the same, so that $\tree{t}$ and $\tree{t}'$ have the same height. Then all the conditions for $\tree{t}' \chaineq_m \tree{t}$ being true have been met.
\end{proof}

\begin{lemma}\label{lemma:useful-lemma-1}
Let $d, p, m \in \mathbb{N}$ be given. For each tree $\tree{t} \in \cl{T}_{d, p}$, if $|\tree{t}| > \xi_{d, p}(m)$, then there is some node $z$ of $\tree{t}$ at height say $h$ in $\tree{t}$, such that (i) the number of children of $z$ in $\tree{t}$ is $> \chi_{h, p}(m)$, and (ii) for any node $y \in \tree{t}_z$ such that $y \neq z$, the number of children of $y$ in $\tree{t}$ (and hence $\tree{t}_z$) is at most $\chi_{h_y, p}(m)$ where $h_y$ is the height of $y$ in $\tree{t}$. (Consequently, $|\tree{t}_y| \leq \xi_{h_y, p}(m)$).
\end{lemma}

\begin{proof}
Let $\tree{t}$ be as in the statement of the lemma. Suppose there is no $z \in \tree{t}$ such that the number of children of $z$ in $\tree{t}$ is $> \chi_{h, p}(m)$ where $h$ is the height of $z$ in $\tree{t}$. Then $|\tree{t}| \leq 1 + \chi_{d, p}(m) + \chi_{d, p}(m) \cdot \chi_{d-1, p}(m) + \ldots + \chi_{d, p}(m) \cdot \chi_{d-1, p}(m) \cdot \ldots \cdot \chi_{1, p}(m) \leq (\chi_{d, p}(m))^{d+1} \leq \xi_{d, p}(m)$, contradicting the premise about size of $\tree{t}$. Then there exists a node $z$ of smallest height $h$ such that the number of children of $z$ is $> \chi_{h, p}(m)$; so that for every node $y \neq z$ in the subtree $\tree{t}_z$ of $\tree{t}$ rooted at $z$, it is the case that the number of children of $y$ in $\tree{t}$ (and hence in $\tree{t}_z$) is at most $\chi_{h_y, p}(m)$ where $h_y$ denotes the height of $y$ in $\tree{t}$. By a similar calculation as above, it follows that $|\tree{t}_y| \leq \xi_{h_y, p}(m)$ for all $y \neq z$ in $\tree{t}_z$.
\end{proof}

\subsection{The downward direction}\label{subsection:ELS-downward}
We first show the downward direction of Theorem~\ref{thm:full-LS-finite}. Our proof is along similar lines as the proof of~\cite[Proposition 6.3]{abhisekh-csl17-arxiv}. We first need the following result akin to~\cite[Lemma 6.2]{abhisekh-csl17-arxiv}.

\begin{lemma}\label{lem:full-LS-helper}
Let $d, p, m  \in \mathbb{N}$ be given. Then for every tree $\tree{t} \in \cl{T}_{d, p}$ such that $|\tree{t}| > \xi_{d, p}(m)$, there exists a tree $\tree{s} \in \cl{T}_{d, p}$ such that (i) $|\tree{t}| - |\tree{s}| \leq \xi_{d-1, p}(m)$, and (ii) $\tree{s} \chaineq_{m} \tree{t}$.
\end{lemma}

\begin{proof}
Let $\tree{t}$ be as in the statement of the lemma. We first observe by Lemma~\ref{lemma:useful-lemma-1} that there is a node $z$ in $\tree{t}$ at height $h \leq d$ such that (i) the number of children of $z$ in $\tree{t}$ is $> \chi_{h, p}(m)$, and (ii) $|\tree{t}_y| \leq \xi_{h-1, p}(m)$ for all children $y$ of $z$ in $\tree{t}$.
Consider $\tree{t}_z$ (which has height equal to $h$); let $\forest{f}$ be the forest of rooted trees obtained by deleting the root of $\tree{t}_z$. Then $\forest{f} = \bigcupdot_{i \in I} \tree{x}_i$ where $\tree{x}_i \in \cl{T}_{h-1, p}$ and $|\tree{x}_i| \leq \xi_{h-1, p}(m)$, and $I$ is of size $> \chi_{h, p}(m)$ (and $I$ could be infinite as well). Let $i^* \in I$ be such that $\tree{x}_{i^*}$ has height equal to $h-1$ (there must be such an $i^*$ since the height of $\tree{t}_z$ is $h$). Then by Lemma~\ref{lem:degree-regulation}(\ref{lem:degree-regulation-downward}), there exists a subset $I'$ of $I$ such that (i) $I'$ contains $i^*$, (ii) $|I'| = |I| - 1$, and (iii) if $\str{A}$ and $\str{A}'$  are resp. the $\mso[m]$ type indicators for the families of trees in $\forest{f}$ and $\forest{f}' = \bigcupdot_{i \in I'} \tree{x}_i$, then $\str{A} \fequiv{q} \str{A}'$ for $q = \rho_{h, p}(m) > \rho_{h-1, p}(m)$. Let $I \setminus I' = \{j\}$ and let $\tree{t}'_z$ be the leaf-hereditary subtree of $\tree{t}_z$ obtained by deleting (the unrooted version of) $\tree{x}_j$ from $\tree{t}_z$; then $\forest{f}'$ is the forest of rooted trees obtained by deleting the root of $\tree{t}'_z$. We observe now that $\tree{t}'_z$ and $\tree{t}_z$ have the same height $h$ (since $\forest{f}'$ contains $\tree{x}_{i^*}$) and further that $\tree{t}'_z \chaineq_m \tree{t}_z$ is indeed true. (The first two conditions of the $\chaineq_m$ definition are already shown satisfied above, and since $\mc{F}' \subseteq \mc{F}$, we use the reflexivity of $\chaineq_m$ from Lemma~\ref{lem:chaineq-props} to see that the last condition of the $\chaineq_m$ definition is satisfied as well). Let $\tree{s}$ be the leaf-hereditary subtree of $\tree{t}$ obtained by replacing (the unrooted version of) $\tree{t}_z$ with (the unrooted version of) $\tree{t}'_z$; then by Lemma~\ref{lemma:simple-FV-comp}, we get that $\tree{s} \chaineq_{m} \tree{t}$. We now observe that $|\tree{t}| - |\tree{s}| = |\tree{t}_z| - |\tree{t}'_z| = |\forest{f}| - |\forest{f}'| = |\tree{s}_j| \leq \xi_{h-1, p}(m) \leq \xi_{d-1, p}(m)$. Then $\tree{s}$ is indeed as desired.
\end{proof}

\begin{proof}[Proof of Theorem~\ref{thm:full-LS-finite}(\ref{thm:full-LS-finite-downward})]

Suppose $\tree{t} \in \cl{T}_{d, p}$ is such that $|\tree{t}| \in \scale{\eta}{\vartheta_{d, p}}$. If $\lambda = \eta$, then taking $\tree{t}' = \tree{t}$ we are done; so assume $\lambda < \eta$.  To prove the theorem, we consider two cases depending on whether $\eta$ is finite or infinite. 

If $\eta$ is finite, then it suffices to show the existence of the desired tree $\tree{t}'$ for $\lambda = \eta - 1$. This is because for all $\mu < \eta$, the  $\chaineq_{\mu}$ relation is transitive and because $\tree{t}' \chaineq_{\eta-1} \tree{t}$ implies $\tree{t}' \chaineq_{\mu} \tree{t}$  (see Lemma~\ref{lem:chaineq-props}). So let $\lambda = \eta -1$. Now since $|\tree{t}| \in \scale{\eta}{\vartheta_{d, p}}$, we have that $|\tree{t}| > \vartheta_{d, p}(\eta) = \xi_{d, p}(\eta) > \xi_{d, p}(\lambda)$. Then by Lemma~\ref{lem:full-LS-helper}, we get that there exists $\tree{s} \in \cl{T}_{d, p}$ such that (i) $|\tree{t}| - |\tree{s}| \leq \xi_{d-1, p}(\lambda)$, and (ii) $\tree{s} \chaineq_{\lambda} \tree{t}$. Now we observe that $|\scale{\lambda}{\vartheta_{d, p}}| = \vartheta_{d, p}(\lambda+1) - \vartheta_{d, p}(\lambda) \ge \xi_{d-1, p}(\lambda)$ by the inequality~(\ref{ineq:vartheta}). Then either $|\tree{s}| \in \scale{\lambda}{\vartheta_{d, p}}$ or $|\tree{s}| \in \scale{\eta}{\vartheta_{d, p}}$. In the former case,  we are done  by taking $\tree{t}' = \tree{s}$ and in the  latter case, we apply Lemma~\ref{lem:full-LS-helper} recursively to $\tree{s}$ and again utilize the transitivity of $\chaineq_{\lambda}$, until we eventually get the desired tree $\tree{t}'$. 

If $\eta$ is infinite, then define $\nu = \vartheta_{d, p}(\lambda+1) + 1$. Observe that $\chi_{d-1, p}(\lambda) < \nu < \eta$. We show below the following statement:

($\dagger$) There exists a finite tree $\tree{t}'' \in \cl{T}_{d, p}$ such that (i) $\nu \leq |\tree{t}''|$, and (ii) $\tree{t}'' \chaineq_{\lambda} \tree{t}$. 

Assuming ($\dagger$), we see that $|\tree{t}''| \in \scale{\eta'}{\vartheta_{d, p}}$ for some finite $\eta' > \lambda$. Then by the previous case above, there exists a tree $\tree{t}' \in \cl{T}_{d, p}$ such that (i) $|\tree{t}'| \in \scale{\lambda}{\vartheta_{d, p}}$, and (ii) $\tree{t}' \chaineq_{\lambda} \tree{t}''$. Then $\tree{t}'$ is as desired since $\chaineq_{\lambda}$ is transitive. We now show ($\dagger$) to complete the proof.

Consider $\tree{t}$; it is such that $|\tree{t}| = \eta$ where $\eta$ is infinite (since  $\scale{\eta}{\vartheta_{d, p}} = \{\eta\}$ for such $\eta$). Then the set $X = \{z \in \tree{t} \mid z~\mbox{has}~ > \nu ~\mbox{many children in}~\tree{t}\}$ is non-empty. (For if not, then every node in $\tree{t}$ has at most $\nu$ children, implying $\tree{t}$ is finite, since it has height $\leq d$.) We first show the existence of $\tree{t}''$ when $\root{\tree{t}} \in X$. Let $\forest{f}$ be the forest of rooted trees obtained by deleting the root of $\tree{t}$; then $\forest{f} = \bigcupdot_{i \in I} \tree{x}_i$ where $\tree{x}_i \in \cl{T}_{d-1, p}$ and $|I| > \nu$. By Corollary~\ref{cor:DLS-chaineq-form}, for each $i \in I$, there exists a tree $\tree{x}'_i \in \cl{T}_{d-1, p}$ such that (i) $|\tree{x}'_i| \leq \zeta_{d-1, p}(\lambda, d-1)$, and (ii) $\tree{x}'_i \chaineq_{\lambda} \tree{x}_i$, and hence $\tree{x}'_i \mequiv{\lambda} \tree{x}_i$ (by Lemma~\ref{lem:chaineq-props}). Then consider the forest $\forest{f}' = \bigcupdot_{i \in I} \tree{x}'_i$. We see that the $\mso[\lambda]$ type indicators $\str{A}$ and $\str{A}'$ resp. of the families of trees in forests $\forest{f}$ and $\forest{f}'$, are isomorphic, and hence $\fo[q]$-equivalent for any $q$. Let the height of $\tree{t}$ be $h \leq d$, and let $i^* \in I$ be such that $\tree{x}_{i^*}$, and hence $\tree{x}'_{i^*}$, has height $h-1$. Since $\nu > \chi_{d-1, p}(\lambda)$, we have by Lemma~\ref{lem:degree-regulation}(\ref{lem:degree-regulation-downward}), that there exists a subset $I'' \subseteq I$ of size exactly $\nu$ and containing $i^*$, such that if $\forest{f}'' = \bigcupdot_{i \in I''} \tree{x}'_i$, then the $\mso[\lambda]$ type indicator $\str{A}''$ of the family of trees in $\forest{f}''$ satisfies $\str{A}'' \fequiv{q} \str{A}' = \str{A}$ for $q = \rho_{d-1, p}(\lambda) \ge \rho_{h-1, p}(\lambda)$.
Then if $\tree{t}''$ is the leaf-hereditary subtree of $\tree{t}$ obtained by deleting (the unrooted versions of) all child subtrees $\tree{x}_i$ for $i \notin I''$, and replacing each $\tree{x}_i$ for $i \in I''$ with $\tree{x}'_i$ (the replacements being of the unrooted versions of the trees), then we see that  $\tree{t}'' \chaineq_{\lambda} \tree{t}$ is true. Further, observe that $\nu \leq |\tree{t}''|$, and that $|\tree{t}''| \leq 1 + \sum_{i \in I''} |\tree{x}_i'| \leq 1 + \nu \cdot \zeta_{d-1, p}(\lambda, d-1) < \omega$, so $\tree{t}''$ is finite. Then $\tree{t}''$ is indeed as desired.

We now consider the case when the $\root{\tree{t}} \notin X$. Let $X' = \{ z \in X \mid z~\mbox{has no ancestor}~z'~\mbox{in}~\tree{t}$ $~\mbox{such that}~z' \in X\}$. By the preceding part, we have for each $z \in X'$, that if $\tree{t}_z$ is the subtree of $\tree{t}$ rooted at $z$, then there is a finite tree $\tree{t}''_z \in \cl{T}_{d-1, p}$ such that (i) $\nu \leq |\tree{t}''_z|`$, and (ii) $\tree{t}''_z \chaineq_{\lambda} \tree{t}_z$. Let $\tree{t}''$ be the leaf-hereditary subtree of $\tree{t}$ obtained by replacing (the unrooted version of) the subtree $\tree{t}_z$ in $\tree{t}$ with (the unrooted version of)  $\tree{t}''_z$ for all $z \in X'$. It is immediate that $\nu \leq |\tree{t}''|$. Further we observe that every node $y$ of $\tree{t}''$ has only finitely many children. This is clear if $y \in \tree{t}''_z$ for some $z \in X'$, since $\tree{t}''_z$ is itself finite. Else, consider a node $y$ of $\tree{t}''$ that is not an element of $\tree{t}''_z$ for any $z \in X'$. If $y$ does not have finitely many children, then $y \in X$ and therefore has an ancestor $y'$ (possibly itself) that is in $X'$. Then $y \in \tree{t}''_{y'}$ (since $y \in \tree{t}''$) contradicting the assumption above about $y$. Thus every node of $\tree{t}''$ has finitely many children; then  $\tree{t}''$ is finite since it has height $\leq d$. Further since $\tree{t}''_z \chaineq_{\lambda} \tree{t}_z$ for all $z \in X'$, we get by repeated applications of Lemma~\ref{lemma:simple-FV-comp} and the transitivity of $\chaineq_\lambda$ from Lemma~\ref{lem:chaineq-props}, that $\tree{t}'' \chaineq_{\lambda} \tree{t}$, completing the proof of ($\dagger$).
\end{proof}

\subsection{The upward direction}\label{subsection:ELS-upward}
We now show the upward direction of Theorem~\ref{thm:full-LS-finite}. \begin{proof}[Proof of Theorem~\ref{thm:full-LS-finite}(\ref{thm:full-LS-finite-upward})]

If $\lambda = \eta$, then taking $\tree{t}' = \tree{t}$ we are done. So assume $\lambda > \eta$.


Consider a tree $\tree{t} \in \cl{T}_{d, p}$ and suppose $|\tree{t}| \in \scale{\eta}{\vartheta_{d, p}}$ for some finite $\eta \ge 1$. Then $|\tree{t}| > \vartheta_{d, p}(\eta) = \xi_{d, p}(\eta)$ and hence by Lemma~\ref{lemma:useful-lemma-1}, there is some node $z$ of $\tree{t}$ at height $h \leq d$ such that (i) the number of children of $z$ in $\tree{t}$ is $> \chi_{h, p}(\eta)$, and (ii) $|\tree{t}_y| \leq \xi_{h-1, p}(\eta)$ for each child $y$ of $z$.
 
Let $\lambda > \eta$ be a cardinal. Consider the subtree $\tree{t}_{z}$ of $\tree{t}$ rooted at $z$ where $z$ is as above. We show below the following statement:

($\ddagger$) There exists a tree $\tree{t}'_{z} \in \cl{T}_{h, p}$ of which $\tree{t}_z$ is a leaf-hereditary subtree such that (i) if $\lambda$ is finite, then $|\tree{t}'_z| - |\tree{t}_z| \leq \xi_{d-1, p}(\eta)$, and if $\lambda$ is infinite, then $|\tree{t}'_z| - |\tree{t}_z| =  \lambda$, and (ii) $\tree{t}_z \chaineq_{\eta} \tree{t}'_z$. 

Then if $\tree{t}''$ is the tree obtained by replacing the subtree $\tree{t}_z$ in $\tree{t}$ with $\tree{t}'_z$ (the replacements being of the unrooted versions of the mentioned trees), then by Lemma~\ref{lemma:simple-FV-comp}, we have that $\tree{t} \chaineq_{\eta} \tree{t}''$. We now examine $|\tree{t}''|$; observe that  $|\tree{t}''| = |\tree{t}'_z| - |\tree{t}_z| + |\tree{t}|$. We have two cases as below:
\begin{enumerate}
    \item If $\lambda$ is finite, then we observe firstly that it suffices to show the theorem for simply $\lambda = \eta + 1$; this is once again due to the transitivity of $\chaineq_{\eta}$ and because $\chaineq_{\mu}$ implies  $\chaineq_{\eta}$ for all $\mu \ge \eta$ (Lemma~\ref{lem:chaineq-props}). So let $\lambda = \eta + 1$. Then by ($\ddagger$), we have $|\tree{t}''| - |\tree{t}| = |\tree{t}'_z| - |\tree{t}_z| \leq \xi_{d-1, p}(\eta)$. Now observe that by inequality~(\ref{ineq:vartheta}), we have $|\scale{\lambda}{\vartheta_{d, p}}| = \vartheta_{d, p}(\lambda + 1) - \vartheta_{d, p}(\lambda) \ge \xi_{d-1, p}(\lambda) \ge \xi_{d-1, p}(\eta)$. Then either $\tree{t}'' \in \scale{\lambda}{\vartheta_{d, p}}$ or $\tree{t}'' \in \scale{\eta}{\vartheta_{d, p}}$. In the former case, we are done by taking $\tree{t}' = \tree{t}''$, and in the latter case, we repeat all of the above arguments with $\tree{t}''$ in place of $\tree{t}$ recursively to eventually obtain $\tree{t}'$ such that $|\tree{t}'| \in \scale{\lambda}{\vartheta_{d, p}}$. The transitivity of $\chaineq_\eta$ again ensures that $\tree{t} \chaineq_\eta \tree{t}'$, as desired. 
    \item If $\lambda$ is infinite, then by ($\ddagger$), $|\tree{t}''| = |\tree{t}| + \lambda = \lambda$, since $\tree{t}$ is finite (as $|\tree{t}| \in \scale{\eta}{\vartheta_{d, p}}$). We are then done by taking $\tree{t}' = \tree{t}''$.
\end{enumerate}

We now show ($\ddagger$) to complete the proof.

Let $\forest{f}$ be the forest of rooted trees obtained by removing the root of $\tree{t}_z$. Then $\forest{f} = \bigcupdot_{i \in I} \tree{s}_i$ where $\tree{s}_i \in \cl{T}_{h-1, p}$ and $|\tree{s}_i| \leq  \xi_{h-1, p}(\eta)$ for all $i \in I$, and $|I| > \chi_{h, p}(\eta)$. By Lemma~\ref{lem:degree-regulation}(\ref{lem:degree-regulation-upward}), for each $\mu > |I|$, there exists $I' \supseteq I$ of size $\mu$ and trees $\tree{s}_i \in \cl{T}_{h-1, p}$ for $i \in I' \setminus I$ such that if $\forest{f}' = \bigcupdot_{i \in I'} \tree{s}_i$, and $\str{A}$ and $\str{A}'$ are resp. the $\mso[m]$ type indicators of $\forest{f}$ and $\forest{f}'$, then $\str{A} \fequiv{q} \str{A}'$ where $q = \rho_{h, p}(\eta) > \rho_{h-1, p}(\eta)$. Further, for all $i \in I' \setminus I$, the tree $\tree{s}_i$ is isomorphic to $\tree{s}_j$ for some $j \in I$. Let now $\tree{t}'_{z, \mu}$ be the extension of $\tree{t}_z$ obtained by adding for each $i \in I' \setminus I$, the (unrooted version of) tree $\tree{s}_i$ as a child subtree of the root of $\tree{t}_z$. So that $\tree{t}_z$ is a leaf-hereditary subtree of $\tree{t}'_{z, \mu}$, and $\forest{f}'$ is exactly the forest of rooted trees obtained by deleting the root of $\tree{t}'_{z, \mu}$. Then $\tree{t}_z$ and $\tree{t}'_{z, \mu}$ have the same height $h$, and further one verifies that $\tree{t}_z \chaineq_{\eta} \tree{t}'_{z, \mu}$.

We now choose $\mu$ suitably to show that ($\ddagger$) is true with $\tree{t}'_z = \tree{t}'_{z, \mu}$. If $\lambda$ is finite, then choose $|I'| = \mu = |I| + 1$. Then $|\tree{t}'_z| - |\tree{t}_z| = |\forest{f}'| - |\forest{f}| = \sum_{i \in I' \setminus I} |\tree{s}_i| = |\tree{s}_j| < \xi_{h-1, p}(\eta) \leq \xi_{d-1, p}(\eta)$. If $\lambda$ is infinite, then choose $|I'| = \mu = \lambda$. Then $|\tree{t}'_z| - |\tree{t}_z| = \sum_{i \in I' \setminus I} |\tree{s}_i| = \lambda \cdot |\tree{s}_j| = \lambda$. 
\end{proof}

\section{An $\mso$ compactness theorem}\label{section:compactness}

For $r, p \in \mathbb{N}$, recall that every tree model $(\tree{s}, S) \in \treemodel{r, p}{d}$ can be seen as a tree in $\cl{T}_{d, r\cdot p + 1}$ via the 1-1 function $f$ that maps the pair $(i, j)$ labeling any leaf node of $\tree{s}$ with $1 \leq i \leq r$ and $1 \leq j \leq p$, to the number $(i-1) \cdot p + j \in [r\cdot p]$; and every internal node of $\tree{s}$ is regarded as having the label $r \cdot p + 1$. Let $\theta$ be an $\mso$ sentence over the vocabulary of $\cl{T}_{d, r\cdot p +1}$ and let $\mc{S} \in \mc{P}([r]^2 \times [p])$ be a set of signatures, where $\mc{P}(X)$ denotes the power set of $X$. We denote by $\treemodel{r, p}{d; \theta, \mc{S}}$ the subclass of $\treemodel{r, p}{d}$ consisting of those tree models $(\tree{s}, S)$ such that $S \in \mc{S}$, and $\tree{s}$, seen as a tree in $\cl{T}_{d, r\cdot p+1}$ as described above, models $\theta$. We correspondingly denote by $\tm{r, p}{d; \theta, \mc{S}}$ the subclass of $\tm{r, p}{d}$ consisting of those graphs that have a tree model in $\treemodel{r, p}{d; \theta, \mc{S}}$. So $\tm{r, p}{d} = \tm{r, p}{d, \true, \mc{P}([r]^2 \times [p]})$. We now have the following compactness property for $\tm{r, p}{d; \theta, \mc{S}}$.

\begin{theorem}\label{thm:shrub-depth-compactness}
Let $d, r, p \in \mathbb{N}$ be given. Let $\theta$ be an $\mso$ sentence over the vocabulary of $\cl{T}_{d, r\cdot p +1}$ and let $\mc{S} \in \mc{P}([r]^2 \times [p])$ be a set of signatures. Then for every set  $T = \{\varphi_1, \varphi_2, \ldots, \}$ of $\mso$ sentences over the vocabulary of $\tm{r, p}{d}$, if every finite subset of $T$ is satisfiable over $\tm{r, p}{d; \theta, \mc{S}}$, then $T$ is satisfiable over $\tm{r, p}{d; \theta, \mc{S}}$. Further $T$ has a countable model in $\tm{r, p}{d; \theta, \mc{S}}$.
\end{theorem}

As with earlier results, Theorem~\ref{thm:shrub-depth-compactness} is established by first showing a corresponding compactness theorem for $\cl{T}_{d, p}$, namely Theorem~\ref{thm:compactness} below, and then transferring the result via $\fo$ interpretations.

\begin{theorem}\label{thm:compactness}
Let $d, p \in \mathbb{N}$ be given. Then for every set  $T = \{\varphi_1, \varphi_2, \ldots, \}$ of $\mso$ sentences over the vocabulary of $\cl{T}_{d, p}$, if every finite subset of $T$ is satisfiable over $\cl{T}_{d, p}$, then $T$ is satisfiable over $\cl{T}_{d, p}$. Further $T$ has a countable model in $\cl{T}_{d, p}$.
\end{theorem}

\begin{proof}[Proof of Theorem~\ref{thm:shrub-depth-compactness}]
Let $\psi_i = \bigwedge_{j = 1}^{j = i} \varphi_i$ for $i \in \mathbb{N}$. By the premise, $\psi_i$ is satisfied in a $p$-labeled graph $G_i \in \tm{r, p}{d; \theta, \mc{S}}$. Let $(\tree{s}_1, S_i) \in \treemodel{r, p}{d}$ be a tree model for $G_i$; then $\tree{s}_i$ seen as a tree of $\cl{T}_{d, r \cdot p + 1}$ models $\theta$, and $S_i \in \mc{S}$ for all $i \ge 1$.  Now since $\mc{S}$ is finite, it must be that there exist infinitely many indices $i_1 < i_2 < \ldots$ such that $S_{i_1} = S_{i_2} = \ldots = S (\mbox{say})$. Then consider the FO interpretation $\Xi_{S, p}$ as described in Section~\ref{section:shrub-depth}; this is such that $G_{i_j} = \Xi_{S, p}(\tree{s}_{i_j})$ for $j \ge 1$.

Consider now the theory $T' = \{\Xi_{S, p}(\varphi_i) \mid i \ge 1\} \cup \{\theta, \Omega_{r, p, d}\}$ and the sentences $\psi_i' = \theta \wedge \Omega_{r, p, d} \wedge \bigwedge_{j = 1}^{j = i} \Xi_{S, p}(\varphi_i) $ for $i \ge 1$, where $\Omega_{r, p, d}$ and $\Xi_{S, p}(\varphi_i)$ are  as defined in Section~\ref{section:shrub-depth}. By the fundamental property of interpretations (cf. Section~\ref{section:shrub-depth}), we have that $G_{i_j} \models \varphi_l$ iff $\tree{s}_{i_j} \models \Xi_{S, p}(\varphi_l)$ for all $l, j \ge 1$. Therefore since $G_{i_j} \models \psi_{i_j}$, we verify that $\tree{s}_{i_j} \models \psi_{i_j}'$. Then every finite subset of  $T'$ is satisfiable in $\cl{T}_{d, r\cdot p +1}$. By Theorem~\ref{thm:compactness}, we have that there is a countable model $\tree{s}^*$ of $T$ in $\cl{T}_{d, r\cdot p+1}$. Since $\tree{s}^* \models \theta \wedge \Omega_{r, p, d}$ and $S \in \mc{S}$, we have that $(\tree{s}^*, S) \in \treemodel{r, p}{d; \theta, \mc{S}}$. Let $G^* = \Xi_{S, p}(\tree{s}^*)$.  Then $G^* \in \tm{r, p}{d; \theta, \mc{S}}$. Finally, since $\tree{s}^* \models \Xi_{S, p}(\varphi_i)$ for all $i \ge 1$, we get by the fundamental property of interpretations again, that $G^* \models \varphi_i$ for all $i \ge 1$; then $G^* \models T$. Observe that since $\tree{s}^*$ is countable, so is $G^*$.
\end{proof}

We now devote ourselves to proving Theorem~\ref{thm:compactness}. 

Define a \emph{chain} of trees of $\cl{T}_{d, p}$ as a sequence $\tree{t}_1, \tree{t}_2, \ldots$ of trees of $\cl{T}_{d, p}$ such that $\tree{t}_i$ is a subtree of $\tree{t}_{i+1}$ for all $i \ge 1$. 
We denote such a chain $\mc{C}$ as $\mc{C} := \tree{t}_1 \subseteq \tree{t}_2 \subseteq \ldots$. We can now define the union $\tree{t}^*$ of $\mc{C}$, denoted $\bigcup_{i \ge 0} \tree{t}_i$, in the natural way as: the universe of $\tree{t}^*$ is the union of the universes of the $\tree{t}_i$'s for $i \ge 1$; and for any $r$-ary predicate $R$ in the vocabulary of $\cl{T}_{d, p}$, the interpretation of $R$ in $\tree{t}^*$ is the union of the interpretations of $R$ in the $\tree{t}_i$'s for $i \ge 1$. It is easy to see that $\tree{t}^*$ belongs to $\cl{T}_{d, p}$. We now have the following lemma that plays a crucial role in the proof of Theorem~\ref{thm:compactness}. 

\begin{lemma}\label{lem:chain-lemma}
Let $d, p \in \mathbb{N}$ be given. Let $\mc{C}$ be an infinite chain of trees of $\cl{T}_{d, p}$ given by $\mc{C} := \tree{t}_1 \subseteq \tree{t}_2 \subseteq \ldots$. Let $\tree{t}^* = \bigcup_{i \ge 1} \tree{t}_i$ be the union of $\mc{C}$. Suppose $\tree{t}_i \chaineq_{m_i} \tree{t}_{i+1}$ for all $i \ge 1$, where $0 \leq m_1 \leq m_2 \leq \ldots$ and $\sup \{m_i  \mid i \ge 1\} = \omega$. Then $\tree{t}_i \chaineq_{m_i} \tree{t}^*$ for all $i \ge 1$. 
\end{lemma}

\begin{proof}[Proof of Lemma~\ref{lem:chain-lemma}]
We prove the lemma by induction on $d$. The base case of $d = 0$ is evident. Assume the statement of the lemma for $d - 1$ where $d \ge 1$. Consider an infinite chain $\mc{C}$ of trees of $\cl{T}_{d, p}$ given by $\mc{C} := \tree{t}_1 \subseteq \tree{t}_2 \subseteq \ldots$ satisfying the conditions in the statement of the lemma.  We first observe that there is no loss of generality in assuming that all the trees of the chain have the same height. For, if the trees have different heights, all of these being $\leq d$ implies that for some $h \leq d$, there is an infinite subchain $\mc{C}'$ of $\mc{C}$ given by $\mc{C}' := \tree{t}_{j_1} \subseteq \tree{t}_{j_2} \subseteq \ldots$ such that $j_1 < j_2 < \ldots$, and all trees in $\mc{C}'$ have height $h$. Then the unions of $\mc{C}'$ and $\mc{C}$ are the same, namely $\tree{t}^*$. Further for $l \ge 1$, since $\tree{t}_{j_l} \chaineq_{m_{j_l}} \tree{t}_{j_l+1} \chaineq_{m_{j_l + 1}} \ldots \tree{t}_{j_{l+1}}$, we infer using Lemma~\ref{lem:chaineq-props} that $\tree{t}_{j_l} \chaineq_{m_{j_l}} \tree{t}_{j_{l+1}}$ for all $l \ge 1$. Observe also that $m_{j_l} \leq m_{j_{l+1}}$ for all $l \ge 1$, and $\sup \{m_{j_l} \mid l \ge 1\} = \omega$. Then assuming we have shown the lemma when all trees of $\mc{C}$ have equal height, we have $\tree{t}_{j_l} \chaineq_{m_{j_l}} \tree{t}^*$. Then for a given $i \ge 1$, let $j_k$ be such that $i \leq j_k$. Then $\tree{t}_i \chaineq_{m_i} \tree{t}_{j_k} \chaineq_{m_{j_k}} \tree{t}^*$; whereby $\tree{t}_i \chaineq_{m_i} \tree{t}^*$ (by Lemma~\ref{lem:chaineq-props}) showing the lemma for the chain $\mc{C}$. So we henceforth assume that all trees of $\mc{C}$ have equal height $h \leq d$.

It is easy to verify that $\tree{t}_i$ is a leaf hereditary subtree of $\tree{t}^*$ for all $i \ge 1$. For each tree $\tree{t}_i$ for $i \ge 1$, consider the forest $\forest{f}_i$ of rooted trees of $\cl{T}_{h-1, p}$ obtained by removing the root of $\tree{t}_i$. Likewise let $\forest{f}^*$ be the forest of rooted trees of $\cl{T}_{h-1, p}$ obtained by removing the root of $\tree{t}^*$; then $\forest{f}^* = \bigcup_{i \ge 1} \forest{f}_i$. Let $\mc{F}_i$ be the family of trees constituting forest $\forest{f}_i$ and $\mc{F}^*$ the family of trees constituting forest $\forest{f}^*$, and let  
$\str{A}_i = \ind{m_i}{\mc{F}_i}$ and $\str{A}^* = \ind{m_i}{\mc{F}^*}$ be resp. the $\mso[m_i]$ type indicators of $\mc{F}_i$ and $\mc{F}^*$. 

Fix an $i \ge 1$. For $j \ge i$, since $\tree{t}_j \chaineq_{m_j} \tree{t}_{j+1}$, we have from Lemma~\ref{lem:chaineq-props} that $\tree{t}_j \chaineq_{m_i} \tree{t}_{j+1}$ since $m_i \leq m_j$. Then for all $j \ge i$, we have that $\str{A}_j \fequiv{q} \str{A}_{j+1}$ for $q = \rho_{h-1, p}(m_i)$ (from the $\chaineq_{m_i}$ definition). We show below that $\str{A}_j \fequiv{q} \str{A}^*$ for all $j \ge i$. We show this by showing that for every predicate $T \in \tau_{m_i, \cl{T}_{h-1, p}}$, if $|T^{\str{A}_j}| \ge q$, then $|T^{\str{A}^*}| \ge q$ as well, and if $|T^{\str{A}_j}| < q$, then $|T^{\str{A}^*}| = |T^{\str{A}_j}|$. Then $\str{A}_j \fequiv{q} \str{A}^*$ follows by a simple $\fo$ EF game argument.

Let $\cl{S}$ be the class of disjoint unions of trees in $\cl{T}_{h-1, p}$. Let $|T^{\str{A}_j}| = n$ and let $\tree{s}_j^l \in \mc{F}_j$ for $1 \leq l \leq n$ be distinct trees such that $\delta_{m_i, \cl{S}}(\tree{s}_j^l) = T$. Since $\tree{t}_i \chaineq_{m_i} \tree{t}_{i+1} \chaineq_{m_i} \tree{t}_{i+2} \chaineq_{m_i} \ldots$, we get that there exist for each $l$ and each $k \ge j$, a unique tree $\tree{s}_{k}^l \in \mc{F}_{k}$ such that $\tree{s}_j^l \chaineq_{m_i} \tree{s}_{j+1}^l \chaineq_{m_i} \tree{s}_{j+2}^l \chaineq_{m_i} \ldots$. Then $T = \delta_{m_i, \cl{S}}(\tree{s}_k^l)$ for all $k \ge j$ and  $1 \leq l \leq n$ (since by Lemma~\ref{lem:chaineq-props}, $\tree{s}_j^l \mequiv{m_i} \tree{s}_{j+1}^l \mequiv{m_i} \ldots$). Further, since $\tree{s}_j^{l_1}$ and $\tree{s}_j^{l_2}$ are distinct for $1 \leq l_1 < l_2 \leq n$, it follows that so are $\tree{s}_k^{l_1}$ and $\tree{s}_k^{l_2}$ for all $k \ge j$. So that $|T^{\str{A}_k}| \ge n$ for all $k \ge j$.

Let's now look at $\forest{f}^*$, $\mc{F}^*$ and $\str{A}^*$. Since 
$\tree{s}_j^l \chaineq_{m_i} \tree{s}_{j+1}^l \chaineq_{m_i} \tree{s}_{j+2}^l \chaineq_{m_i} \ldots$ for $1 \leq l \leq n$, and all trees in the chain belong to $\cl{T}_{h-1, p}$ (they could be of different heights), we get by induction hypothesis that $\tree{s}_j^l \chaineq_{m_i} \tree{s}^*_l = \bigcup_{k \ge j} \tree{s}_k^l$. Then $\delta_{m_i, \cl{S}}(\tree{s}^*_l) = \delta_{m_i, \cl{S}}(\tree{s}_j^l) = T$ for all $l \in \{1, \ldots, n\}$. Then $|T^{\str{A}^*}| \ge n$ since clearly $\tree{s}^*_{l_1}$ and $\tree{s}^*_{l_2}$ are distinct for $1 \leq l_1 < l_2 \leq n$. Thus if 
$|T^{\str{A}_j}| = n \ge q$, then $|T^{\str{A}^*}| \ge q$. Suppose $n < q$  
but $|T^{\str{A}^*}| > n$. Then for some $\tree{z}^* \in \mc{F}^*$, it holds that $\tree{z}^*$ is distinct from $\tree{s}^*_l$ for all $l \in \{1, \ldots, n\}$ and  $\delta_{m_i, \cl{S}}(\tree{z}^*) = T$.  Now since $\tree{z}^* \in \mc{F}^*$, there exists $j^* \ge j$ and trees $\tree{z}_k \in \mc{F}_k$ for all $k \ge j^*$ such that $\tree{z}_k$ is a leaf-hereditary subtree of $\tree{z}_{k+1}$ and $\tree{z}^* = \bigcup_{k \ge j^*} \tree{z}_k$.  Since $\tree{t}_k \chaineq_{m_i} \tree{t}_{k+1}$, we must have $\tree{z}_k \chaineq_{m_i} \tree{z}_{k+1}$ from the definition of $\chaineq_{m_i}$, for all $k \ge j^*$. Given that each tree $\tree{z}_k$ belongs to $\cl{T}_{d-1, p}$, we obtain by our induction hypothesis that $\tree{z}_k \chaineq_{m_i} \tree{z}^*$ for each $k \ge j^*$. Then $\tree{z}_k \mequiv{m_i} \tree{z}^*$ (by Lemma~\ref{lem:chaineq-props}) and hence $\delta_{m_i, \cl{S}}(\tree{z}_k) = T$ for all $ k \ge j^*$ and hence in particular $\delta_{m_i, \cl{S}}(\tree{z}_{j^*}) = T$ . Now since $\tree{t}_j \chaineq_{m_i} \tree{t}_{j+1} \chaineq_{m_i} \ldots \chaineq_{m_i} \tree{t}_{j^*}$, we have by the transitivity of $\chaineq_{m_i}$ that $\tree{t}_j \chaineq_{m_i} \tree{t}_{j^*}$ and hence by the definition of  $\chaineq_{m_i}$, we get that $\str{A}_j \fequiv{q} \str{A}_{j^*}$. Since $|T^{\str{A}_j}| = n < q$, it follows that  $|T^{\str{A}_{j^*}}| = n$ as well. But we now observe that
the trees $\tree{s}_{j^*}^l$ for $1 \leq l \leq n$ and $\tree{z}_{j^*}$ are all distinct, and yet $\delta_{m_i, \cl{S}}(\tree{s}_{j^*}^l) = T = \delta_{m_i, \cl{S}}(\tree{z}_{j^*})$ for  $1 \leq l \leq n$, which implies $|T^{\str{A}_{j^*}}| > n$ -- a contradiction. Then $|T^{\str{A}^*}| = n$.

In summary, we have shown so far conditions (i) and (ii) of the $\tree{t}_j \chaineq_{m_i} \tree{t}^*$ definition. To see the last condition, suppose $\tree{y}_j \in \mc{F}_j$. Then by a similar reasoning as above, there exist trees $\tree{y}_k \in \mc{F}_k$ for $k \ge j$ such that $\tree{y}_j \chaineq_{m_i} \tree{y}_{j+1} \chaineq_{m_i} \tree{y}_{j+2} \chaineq_{m_i} \ldots$. Once again, reasoning as above, the tree $\tree{y}^* = \bigcup_{k \ge j} \tree{y}_k$ is such that $\tree{y}^* \in \mc{F}^*$ and $\tree{y}_j \chaineq_{m_i} \tree{y}^*$. Thus $\tree{t}_j \chaineq_{m_i} \tree{t}^*$ for all $j \ge i$ and hence in particular $\tree{t}_i \chaineq_{m_i} \tree{t}^*$.
\end{proof}

\begin{proof}[Proof of Theorem~\ref{thm:compactness}]
Let $\psi_i = \bigwedge_{j = 1}^{j = i} \varphi_j$ for $i \in \mathbb{N}$. By the premise, we know that $\psi_i$ is satisfiable for all $i$. Denote by $m_i$, the rank of $\psi_i$ for all $i \ge 1$; then $1 \leq m_1 \leq m_2 \leq \ldots$. If $\sup \{m_i \mid i \ge 1\} < \omega$, then for some $i^* \ge 1$, it holds that $m_j = m_{i^*}$ for all $j \ge i^*$. Then $T$ is a set of MSO formulae of rank at most $m_{i^*}$, and is hence finite upto equivalence. Then $T$ is already satisfiable in $\cl{T}_{d, p}$ since every finite subset of it is satisfiable in $\cl{T}_{d, p}$ by assumption. Likewise, if $\psi_i$ has only finitely many models for any $i \ge 1$, then, since $\psi_{j+1}$ implies $\psi_j$ for all $j \ge i$, it holds that for some $j^* > i$ and for all $k \ge j^*$, the sentence $\psi_k$ is equivalent to the (satisfiable) sentence $\psi_{j^*}$. Then any model of $\psi_{j^*}$ models all of $T$. We therefore assume below that $\sup \{m_i \mid i \ge 1\} = \omega$, and that $\psi_i$ has infinitely many models for each $i \ge 1$. 

Recall the function $\vartheta_{d, p}: \mathbb{N}_+ \rightarrow \mathbb{N}$ that witnesses Theorem~\ref{thm:full-LS-finite}. We observe that for all $i \ge 1$ that $\vartheta_{d, p}(i) \leq \vartheta_{d, p}(i+1)$. Let $\cl{X}_i$ for $i \ge 1$ be the class of all finite trees $\tree{t} \in \cl{T}_{d, p}$ such that (i) $|\tree{t}| \in \scale{m_i}{\vartheta_{d, p}}$, and (ii) $\tree{t} \models \psi_i$. Observe that $\cl{X}_i$ is a finite class (upto isomorphism). We now claim the following for each $i \ge 1$:

\begin{enumerate}[(C1)]
    \item  For every $\tree{t} \in \cl{X}_{i+1}$, there exists $\tree{s} \in \cl{X}_i$ such that  $\tree{s} \chaineq_{m_i} \tree{t}$.\label{claim:1} 
    \item  If $\cl{X}_i$ is empty, then $\psi_i$ has only finitely many models in $\cl{T}_{d, p}$.\label{claim:2}
\end{enumerate}

Let us prove the two claims in order. We recall from above that for any $\tree{t} \in \cl{X}_{i+1}$, we have $|\tree{t}| \in \scale{m_{i+1}}{\vartheta_{d, p}}$. Since $m_{i+1} \ge m_i$, we have by Theorem~\ref{thm:full-LS-finite}(\ref{thm:full-LS-finite-downward}) that there exists a $\tree{s} \in \cl{T}_{d, p}$ such that (i) $|\tree{s}| \in \scale{m_i}{\vartheta_{d, p}}$, and (ii) $\tree{s} \chaineq_{m_i} \tree{t}$. Now by the latter and Lemma~\ref{lem:chaineq-props}, we have $\tree{s} \mequiv{m_i} \tree{t}$. Since $\tree{t} \in \cl{X}_{i+1}$, we have $\tree{t}$ models $\psi_{i+1}$ and hence $\psi_i$ (the latter being of rank $m_i$). Then $\tree{s}$ models $\psi_i$ and since $|\tree{s}| \in \scale{m_i}{\vartheta_{d, p}}$, we have $\tree{s} \in \cl{X}_i$. This shows~\ref{claim:1}. For~\ref{claim:2}, suppose that $\psi_i$ has infinitely many models in $\cl{T}_{d, p}$. Then by Theorem~\ref{thm:full-LS-finite}(\ref{thm:full-LS-finite-downward}), there exists a model $\tree{t}$ of $\psi_i$ in $\cl{T}_{d, p}$ such that $|\tree{t}| \in \scale{m_i}{\vartheta_{d, p}}$; in other words, $\tree{t} \in \cl{X}_i$ showing~\ref{claim:2} in the contrapositive.

Since at the outset itself, we have assumed that all $\psi_i$'s have infinitely many models, we conclude that $\cl{X}_i \neq \emptyset$ for all $i \ge 1$. Further we have already observed above that $\cl{X}_i$ is finite for all $i \ge 1$. We now show the following.

$(\ddagger)$: There exists an infinite chain $\mc{C}$ given by $\mc{C} := \tree{t}_1 \subseteq \tree{t}_2 \subseteq \ldots$ of trees in $\cl{T}_{d, p}$ such that for all $i \ge 1$, (i) $\tree{t}_i \in \cl{X}_i$, and (ii) $\tree{t}_i \chaineq_{m_i} \tree{t}_{i+1}$.

We claim that $(\ddagger)$ implies that $T$ is satisfiable over  $\cl{T}_{d, p}$. Recalling that $\sup \{m_i \mid i \ge 1\} = \omega$ (as assumed at the outset), we have by  $(\ddagger)$ and Lemma~\ref{lem:chain-lemma}, that if $\tree{t}^* = \bigcup_{i \ge 1} \tree{t}_i$ is the union of the chain $\mc{C}$, then $\tree{t}_i \chaineq_{m_i} \tree{t}^*$ for all $i \ge 1$. Then by Lemma~\ref{lem:chaineq-props}, we have $\tree{t}_i \mequiv{m_i} \tree{t}^*$ for all $i \ge 1$. Since $\tree{t}_i \in \cl{X}_i$, we have $\tree{t}_i \models \psi_i$ and hence $\tree{t}^* \models \psi_i$ as well since $\mbox{rank}(\psi_i) = m_i$. Since this holds for all $i \ge 1$, we have $\tree{t}^* \models T$. Further since $\tree{t}_i$ is finite for each $i \ge 1$, it follows that $\tree{t}^*$ is countable.

We now show $(\ddagger)$ to complete the proof. Let $\mathsf{iso}(\cl{X}_i)$ denote the finite set of isomorphism classes of the trees of $\cl{X}_i$ for $i \ge 1$. 

We construct a labeled tree $V$ with countably many ``levels". The root of the tree is at level 0 and for $i \ge 1$, the $i^{\text{th}}$ level contains $|\mathsf{iso}(\cl{X}_i)|$ many nodes. Each node at level $i$ is labeled with exactly one element of $\mathsf{iso}(\cl{X}_i)$ and no two nodes at level $i$ get the same label. (It is possible though that the same label appears in different levels.) This defines the vertices and their labeling. The edge relation is now defined as follows. Firstly the relation relates only pairs of vertices in adjacent levels. All nodes of level 1 are adjacent to the root. For a node $u$ at level $i+1$ for $i \ge 1$, let the isomorphism class of $\tree{t} \in \cl{X}_{i+1}$ be the label of $u$. By Claim~\ref{claim:1}, there exists some node $v$ at level $i$ and some tree $\tree{s} \in \cl{X}_i$ such that the isomorphism class of $\tree{s}$ is the label of $v$, and $\tree{s} \chaineq_{m_i} \tree{t}$. We choose exactly one such node $v$ and put an edge between $u$ and $v$. The edges as specified above are the only edges in $V$; it follows then that $V$ is a tree.

We now observe that $V$ is infinite (since $\cl{X}_i$ is non-empty for all $i$) and finitely branching (since $\mathsf{iso}(\cl{X}_i)$ is finite for all $i$). Then by K\"onig's lemma, there exists an infinite path $u_0, u_1, u_2, \ldots$ in $V$ starting at the root $u_0$. Let the isomorphism class of $\tree{t}_i$ be the label of $u_i$ for $i \ge 1$. Then from the properties of the edge relation, we have for all $i \ge 1$, that (i) $\tree{t}_i \in \cl{X}_i$, and (ii) $\tree{t}_i \chaineq_{m_i} \tree{t}_{i+1}$. Indeed then $\mc{C} := \tree{t}_1 \subseteq \tree{t}_2 \subseteq \ldots$ is the desired chain.
\end{proof}

\section{$\mso = \fo$ over bounded shrub-depth classes}\label{section:mso=fo}
\newcommand{\fancy}[1]{\ensuremath{\mathcal{#1}}}
We show the following in this section. This result for $p = 1$ is already known from~\cite{CF20}, and for $p = 1$ and finite graphs, is known earlier still from~\cite{shrub-depth-FO-equals-MSO}.

\begin{theorem}\label{thm:mso=fo}
Let $r, p, d \in \mathbb{N}$ be given. Then over $\tm{r, p}{d}$, every $\mso$ sentence is equivalent to an $\fo$ sentence. Consequently, the same holds over any subclass of $\tm{r, p}{d}$.
\end{theorem}

We need the following notion for the proof. Given any $\mso$ formula $\theta(\bar{x})$, and an $\fo$ formula $\gamma(y)$, we can construct a \emph{relativized} $\mso$ formula $\theta|_{\gamma}(\bar{x})$ that relativizes all quantifiers of $\theta$ to range over only elements $y$ satisfying $\gamma(y)$. Specifically, to obtain $\theta|_\gamma(\bar{x})$, we replace every subformula of the form $\exists v \psi(v, \bar{u})$ in $\theta(\bar{x})$ with $\exists v (\gamma(v) \wedge \psi|_\gamma(v, \bar{u}))$
and every subformula of the form $\forall v \psi(v, \bar{u})$ with $\forall v (\gamma(v) \rightarrow \psi|_\gamma(v, \bar{u}))$. The key property of relativization is that for any given structure $\str{A}$, if $\str{A}_\gamma$ is the substructure of $\str{A}$ induced by the set $A' = \{b \in \str{A} \mid \str{A} \models \gamma(b)\}$, then for an $\mso$ sentence $\theta$, we have that $\str{A} \models \theta|_\gamma$ iff $\str{A}|_\gamma \models \theta$. In addition to the mentioned notion, we will need the following result from classical model theory~\cite{lindstrom}. 

\begin{theorem}\label{thm:elem-equi-and-iso}
Let $\str{A}$ and $\str{B}$ be two countable structures over the same vocabulary. If $\str{A} \fequiv{n} \str{B}$ for all $n \ge 1$, then $\str{A} \cong \str{B}$.
\end{theorem}

\begin{proof}[Proof of Theorem~\ref{thm:mso=fo}]
Our proof is an adaptation of the proof of Lindstr\"om's theorem (mentioned at the outset in Section~\ref{section:intro}) as presented in~\cite{vaananen}. We prove Theorem~\ref{thm:mso=fo} by contradiction. Suppose there is an $\mso$ sentence $\alpha$ over the vocabulary $\tau$  of $\tm{r, p}{d}$ that is not equivalent over $\tm{r, p}{d}$ to any $\fo$ sentence. Then for every $n \in \mathbb{N}$, there exist graphs $G_n, H_n \in \tm{r, p}{d}$ such that (i) $G_n \fequiv{n} H_n$, and (ii) $G_n \models \alpha$ but $H_n \models \neg \alpha$.  Let $H'_n$ be the graph obtained from $H_n$ be changing every vertex label $j \in [p]$ to $j+p \in [2p]$. This graph is in $\tm{r, 2p}{d}$; take any tree model if $(\tree{t}_1, S_1) \in \treemodel{r, p}{d}$ is a tree model for $H_n$, then changing every leaf label of the form $(i, j)$ in $\tree{t}_1$ to $(i, j+p)$ gives a tree model $(\tree{t}_1', S_1) \in \treemodel{r, 2p}{d}$ for $H'_n$. Let $(\tree{t}_2, S_2) \in \treemodel{r, p}{d}$ be a tree model for $G_n$; then $(\tree{t}_2, S_2) \in \treemodel{r, 2p}{d}$ as well. Consider now the tree model $(\tree{s}, S) \in \treemodel{r, 2p}{d+1}$ where $\tree{s}$ is obtained by making (the unrooted versions of) $\tree{t}_2$ and $\tree{t}'_1$ as the (only) child subtrees of a new root node, and $S = S_1 \cup S_2$. It is easily verified that $(\tree{s}, S)$ is a tree model for the disjoint union $\fancy{M}_n = G_n \cupdot H_n'$ of $G_n$ and $H_n'$. Observe that $S \subseteq [r]^2 \times [d]$.

Consider now the $\fo$ sentence $\theta$ over the vocabulary of $\cl{T}_{d+1, r \cdot 2p + 1}$ that asserts over the subclass $\treemodel{r, 2p}{d+1}$ of $\cl{T}_{d+1, r \cdot 2p + 1}$ that the root has exactly 2 children $x_1$ and $x_2$, and any leaf node whose label corresponds  
(via the bijective function $f: [r] \times [2p] \rightarrow [r \cdot 2p]$ seen in Section~\ref{section:shrub-depth}) to (i) the pair $(i, j)$ where $i \in [r]$ and $j \in [p]$ has a path of length exactly $d$ to $x_1$, and (ii) the pair $(i, j)$ where $i \in [r]$ and $j \in \{p+1, \ldots, 2p\}$ has a path of length exactly $d$ to $x_2$. Let $\mc{S} = \mc{P}([r]^2 \times [d])$. Then any graph in $\tm{r, p}{d+1; \theta, \mc{S}}$ is such that it is a disjoint union of two graphs say $A$ and $B$ where $A \in \tm{r, p}{d}$ and $B$ is "affinely" in $\tm{r, p}{d}$ in that, the labels of its vertices belong to $\{p+1, \ldots, 2p\}$ and the graph $B'$ obtained by changing every label $j$ of a vertex in $B$ to $j - p$, satisfies $B' \in \tm{r, p}{d}$. We see that for each $n \ge 1$, the $2p$-labeled graph $\fancy{M}_n$ indeed belongs to $\tm{r, 2p}{d+1; \theta, \mc{S}}$. 

We now code the properties of $G_n$ and $H_n$ mentioned above into an $\mso$ sentence $\beta_n$ over the vocabulary $\tau = \{E, P_1, \ldots, P_{2p}\}$ of $\tm{r, 2p}{d+1}$. To be able to do this, define for an $\mso$ sentence $\psi$ over the vocabulary $\{E, P_1, \ldots, P_p\}$ of $\tm{r, p}{d}$, the "affine" $\mso$ sentence $\psi[i \mapsto i+p]$ which is obtained from $\psi$ by replacing every atomic predicate of the form $P_i(z)$ with $P_{i+p}(z)$ for $1 \leq i \leq p$. This sentence is over the vocabulary $\{E, P_{p+1}, \ldots, P_{2p}\} \subseteq \tau$. We are now ready to define $\beta_n$. Below $\Delta$ is a notational short hand for $\fdelta{n}{\tm{r, p}{d}}$. Recall also from Section~\ref{section:shrub-depth}, the sentence $\Theta_\delta$ that defines an equivalence class $\delta \in \Delta$. Recall also the sentence $\alpha$ from the outset.
\[
\begin{array}{lll}
    \gamma_1(x) & := & \bigvee\limits_{i = 1}^{i = p} P_i(x)~~~~;~~~~ 
    \gamma_2(x) :=  \bigvee\limits_{i = p+1}^{i = 2p} P_i(x)\\
    \beta_n & := & \exists x \gamma_1(x) \wedge \exists x \gamma_2(x) \wedge \\
    & & \alpha|_{\gamma_1} \wedge (\neg \alpha[i \mapsto i+p])|_{\gamma_2} \wedge \bigwedge\limits_{\delta \in \Delta} \Theta_\delta|_{\gamma_1} \leftrightarrow \Theta_\delta[i \mapsto i+p]|_{\gamma_2} \\
\end{array}
\]

So $\beta_n$ asserts that over $\tm{r, 2p}{d+1; \theta, \mc{S}}$, a $2p$-labeled graph, which is a disjoint union of two non-empty graphs $A \in \tm{r, p}{d}$ and $B$ that is affinely in $\tm{r, p}{D}$ as seen above, is such that (i) $A$ models $\alpha$ but $B$ (seen as a $p$-labeled graph) does not, and (ii)  $A \fequiv{n} B$ (where $B$ is seen as a $p$-labeled graph). Clearly, $\fancy{M}_n \models \beta_n$. Further observe that $\beta_n$ implies $\beta_{n'}$ for $n \ge n'$.

Let $T = \{\beta_n \mid n \ge 1\}$. Then every finite subset of $T$ is satisfiable $\tm{r, 2p}{d+1; \theta, \mc{S}}$ since if $k$ is the largest index such that  $\beta_k$ appears in the subset, then the subset is equivalent to $\beta_k$ and is hence satisfied in $\fancy{M}_k \in \tm{r, 2p}{d+1; \theta, \mc{S}}$. By Theorem~\ref{thm:shrub-depth-compactness}, we get that $T$ is satisfied in a countable model $\fancy{M}^* \in \tm{r, 2p}{d+1; \theta, \mc{S}}$. Then by our observations above, $\fancy{M}^*$ is a disjoint union of two graphs $G^*$ and $H^{**}$ such that $G^* \in \tm{r, p}{d}$ and $H^{**}$ is affinely in $\tm{r, p}{d}$, that is the graph $H^*$ obtained by changing every label $j \in \{p+1, \ldots, 2p\}$ of a vertex in $H^{**}$, to $j - p$ is such that $H^* \in \tm{r, p}{d}$. Since $\fancy{M}^* \models T$, we get that (i) $G^* \models \alpha$ and $H^{**} \models \neg \alpha[i\mapsto i+p]$, so $H^* \models \neg \alpha$, and (ii) $G^* \fequiv{n} H^*$ for all $n \ge 1$. But now since $\fancy{M}^*$ is countable, so are $G^*$ and $H^*$, and therefore by Theorem~\ref{thm:elem-equi-and-iso}, we have $G^* \cong H^*$ -- this is a contradiction since $G^*$ models $\alpha$ but $H^*$ does not.
\end{proof}

\section{Conclusion}\label{section:conclusion}
\renewcommand{\cmso}{\ensuremath{\mathrm{CMSO}}}
In this paper, we looked at classes of arbitrary graphs of bounded shrub-depth, specifically the class $\tm{r, p}{d}$ of $p$-labeled arbitrary graphs whose underlying unlabeled graphs have tree models of height $d$ and $r$ labels. We showed that this class satisfies an extension of the classical {\ls} (\lsref) property into the finite and for $\mso$. This extension being a generalization of the small model property, shows the pseudo-finiteness of the graphs of $\tm{r, p}{d}$. In addition, we obtain as consequences entirely new proofs of various known results concerning bounded shrub-depth classes (of finite graphs) and $\tm{r, p}{d}$. These include the small model property for $\mso$ with elementary bounds, the classical compactness theorem from model theory over $\tm{r, p}{d}$ (known because $\tm{r, p}{d}$ is an elementary class), and the equivalence of $\mso$ and $\fo$ over $\tm{r, p}{d}$ and hence over bounded shrub-depth classes. We prove the last of these by adapting the proof of the classical Lindstr\"om's theorem characterizing $\fo$ over arbitrary structures.

An interesting consequence of our results is that they allow for transferring results back and forth between finite bounded shrub-depth graphs and and their infinite counterparts. For instance, one can obtain a different proof of the $\mso$-$\fo$ equivalence over $\tm{r, 1}{d}$ by simply lifting this equivalence known over bounded shrub-depth classes of finite graphs, to $\tm{r, 1}{d}$. Similarly, one can show that $\tm{r, 1}{d}$ has only finitely many forbidden finite induced subgraphs, where these subgraphs are the same as those constituting a forbidden induced subgraph characterization (in the finite) of the subclass of finite graphs of $\tm{r, 1}{d}$~\cite{shrub-depth-definitive}. (Indeed, the mentioned subgraphs also characterize all of $\tm{1, p}{d}$ over arbitrary structures, since this class is known to be defined by a finite set of forbidden finite induced subgraphs~\cite{CF20}.)
Finally, an example of an infinite to finite transfer is classical preservation theorems, such as the {\lt} theorem, that is true of $\tm{r, p}{d}$ since it is an elementary class, and hence true of the class of finite graphs of $\tm{r, p}{d}$; this is because any $\fo$ sentence that defines a hereditary subclass of the latter class also defines a hereditary subclass of the former owing to the fact that the  graphs of $\tm{r, p}{d}$ are pseudo-finite.

For future work, we believe that similar results as in the paper can be obtained for guarded second order logic $(\gso)$ and infinite graphs of bounded tree-depth, since the Feferman-Vaught composition theorem of~\cite{FO-MSO-coincide} that we use as our main tool in the paper, has also been shown for $\gso$. In particular therefore, it might be possible to get an alternative Lindstr\"om-style proof as above, of the equivalence of $\gso$ and $\fo$ over finite graphs of bounded tree-depth~\cite{FO-MSO-coincide}. We would also like to invesigate extensions of our results to the extension $\cmso$ of $\mso$ with modulo counting quantifiers. Finally, we would like to explore what further results from classical model theory can be transferred to the finite for bounded shrub-depth classes and conversely, what finitary results can be transported to the infinite.

\vspace{5pt}\noindent\tbf{Acknowledgements:} I would like to thank Pascal Schweitzer for discussions pertaining to the $\lelsref$ property. I also thank Anand Pillay for pointing me to~\cite{vaananen} for the proof of Lindstr\"om's theorem.

\bibliographystyle{plain}
\bibliography{refs}

\end{document}